\documentclass[onefignum,onetabnum]{siamsials251208}

\usepackage{lipsum}
\usepackage{amsfonts}
\usepackage{amsopn}
\usepackage{amssymb}
\usepackage{graphicx}
\usepackage{epstopdf}
\usepackage{algorithm}
\usepackage{algpseudocode}
\usepackage{tikz}
\usepackage{multirow}
\ifpdf
  \DeclareGraphicsExtensions{.eps,.pdf,.png,.jpg}
\else
  \DeclareGraphicsExtensions{.eps}
\fi

\usepackage{enumitem}
\setlist[enumerate]{leftmargin=.5in}
\setlist[itemize]{leftmargin=.5in}

\newsiamremark{remark}{Remark}
\newsiamremark{hypothesis}{Hypothesis}
\crefname{hypothesis}{Hypothesis}{Hypotheses}
\newsiamthm{claim}{Claim}
\newsiamremark{fact}{Fact}
\crefname{fact}{Fact}{Facts}

\ifpdf
\hypersetup{
  pdftitle={EDA Decomposition by GMS},
  pdfauthor={X. Chen et al}
}
\fi

\usepackage{xspace}

\newcommand{\vb}{\mathbf{b}}

\newcommand{\ve}{\mathbf{e}}
\newcommand{\vf}{\mathbf{f}}

\newcommand{\vh}{\mathbf{h}}

\newcommand{\vp}{\mathbf{p}}

\newcommand{\vv}{\mathbf{v}}
\newcommand{\vw}{\mathbf{w}}
\newcommand{\vx}{\mathbf{x}}
\newcommand{\vy}{\mathbf{y}}
\newcommand{\vz}{\mathbf{z}}

\newcommand{\mA}{\mathbf{A}}
\newcommand{\mB}{\mathbf{B}}
\newcommand{\mC}{\mathbf{C}}
\newcommand{\mD}{\mathbf{D}}
\newcommand{\mE}{\mathbf{E}}

\newcommand{\mG}{\mathbf{G}}
\newcommand{\mH}{\mathbf{H}}
\newcommand{\mI}{\mathbf{I}}

\newcommand{\mL}{\mathbf{L}}
\newcommand{\mM}{\mathbf{M}}

\newcommand{\mP}{\mathbf{P}}

\newcommand{\mR}{\mathbf{R}}
\newcommand{\mS}{\mathbf{S}}

\newcommand{\mU}{\mathbf{U}}
\newcommand{\mV}{\mathbf{V}}
\newcommand{\mW}{\mathbf{W}}
\newcommand{\mX}{\mathbf{X}}
\newcommand{\mY}{\mathbf{Y}}

\newcommand{\R}{\mathbb{R}}
\newcommand{\Xc}{\mathcal{X}}
\newcommand{\Tc}{\boldsymbol{{\mathcal{T}}}}
\newcommand{\Bc}{\mathcal{BC}}
\newcommand{\Bj}{\mathcal{BJ}}
\newcommand{\OR}{{\rm{OR}}}
\newcommand{\rank}{{\rm{rank}}}

\newcommand{\figpath}{plots/}

\DeclareMathOperator*{\argmin}{argmin}

\DeclareMathOperator{\diag}{diag}

\headers{EDA Decomposition by GMS}{X. Chen et al}

\title{gmsEDA: Decomposition of Electrodermal Activity Signals Using Matrix Separation
\thanks{X. Chen is partially funded by NSF DMS-2307827. X. Chen, D. MacQueen, W. D. Washington, and M. Lammers  are supported by UNCW Brain Health Resilience Hub.
}
}

\author{Xuemei Chen\thanks{Department of Mathematics and Statistics, University of North Carolina Wilmington 
  (\email{chenxuemei@uncw.edu}, \email{lammersm@uncw.edu}).}
\and David MacQueen\thanks{Department of Psychology, University of North Carolina Wilmington
  (\email{macqueend@uncw.edu}, \email{donlinw@uncw.edu}, \email{sec2431@uncw.edu}, \email{ml6796@uncw.edu}).}
\and Wendy Donlin Washington\footnotemark[3]
\and Mark Lammers\footnotemark[2]
\and Owen Deen\thanks{Department of Mathematics, University of Maryland (\email{odeen@umd.edu})}
\and Sean Carey \footnotemark[3]
\and Margot Ledford \footnotemark[3]}

\begin{document}

\maketitle

\begin{abstract}
Electrodermal activity (EDA) signals, which reflect sympathetic nervous system arousal through changes in skin conductance, are widely used in psychological and behavioral research. Decomposing an observed EDA signal into its slowly varying tonic baseline and stimulus-driven phasic component  is an important preprocessing step; however, existing methods process signals in isolation and remain highly sensitive to noise and motion artifacts. 
This work introduces gmsEDA, a new  decomposition method based on generalized matrix separation whose model is designed to cope with noise and motion artifacts. Our method analyzes multiple recordings jointly rather than one at a time, taking advantage of patterns shared across signals to produce more accurate and robust results. Numerical experiments on both simulated and real data  shows that this approach outperforms existing standard tools.
\end{abstract}

%



\section{Introduction}

Electrodermal activity (EDA), also known as galvanic skin response (GSR), reflects the changes in the skin's electrical properties due to activity of the sweat gland controlled by the body's sympathetic nervous system. The sweat glands are particularly concentrated in the palms and are a crucial part of the sympathetic nervous system. The activation of these glands leads to sweat secretion, which changes the skin's electrical conductance. EDA is a widely used index of autonomic arousal caused by behavioral, cognitive, and emotional processes. EDA measurements are considered to be useful in studying stress and anxiety in comparison to other physiological measurements such as heart rate, respiration rate and skin temperature~\cite{cacioppo2007}.
There has been extensive work on evaluating the association between EDA and stress detection, emotional state~\cite{jaques2015predicting, rahma2022electrodermal}, reaction to video content~\cite{eda_cs}, etc. We refer interested readers to \cite{tronstad2022current} and the references therein for a more extensive read.

More recently, research-grade wearable or mobile devices such as Empatica E4 \cite{E4} and Shimmer3 GSR+ \cite{Shimmer} have been widely adopted in clinical and research settings~\cite{evonko2024opiaid}. Such devices allow the vast expansion of data collection, but can create artifacts in the measurement due to movements. New insights may be gained by the synchrony of physiological measurements between multiple subjects over time~\cite{stuldreher2020physiological}. Van Beers et al. \cite{van2020comparison} compared laboratory and wearable sensors (ECG and EDA) in the context of physiological synchrony, and found no significant difference in classification accuracies between the laboratory and wearable sensors. 

An EDA signal consists of two primary components: a slowly varying \textit{tonic} component and a rapidly varying \textit{phasic} component. The tonic is a baseline for the skin's conductance, while the phasic component relates to the discrete changes as a response to an internal or external stimulus.
An observed EDA signal can be viewed as a superposition of the  tonic component (baseline), the phasic component, and noise (eg. generated by measurement devices):
\begin{equation} \label{equ:model1}
    \vy = \vb + \vp + \ve,
\end{equation}
where $\vy$ is the observed EDA signal, $\vb$ is the baseline, $\vp$ is the phasic component, and $\ve$ is the noise component. 

Decomposition of the EDA signal into the tonic and phasic components is considered as a crucial signal processing step, as well as a challenging task~\cite{BK10, cvxEDA, chaspari2014sparse, eda_cs, veeranki2024comparison, tsirmpas2025transformer}. 
Such preprocessing often improves emotion detection~\cite{zhu2024emotion} and other downstream tasks such as opioid withdrawal detection~\cite{evonko2024opiaid}.
The phasic component is a result of the Skin Conductance Response (SCR) events (such as user excitement events).
In this paper, we will model the phasic component as a linear time-invariant system~\cite{A05, BK10, gerster2018testing} where it is a convolution of the SCR events signal $\vx$ and the impulse response signal $\vh$ as $\vp=\vh*\vx$.
Combined with \eqref{equ:model1}, our signal model looks like
\begin{equation}\label{equ:model}
\vy=\vb+\vh*\vx+\ve,
\end{equation}
where $\vx$ can be considered as the indicator of SCR events, a \textbf{sparse} signal.
This signal model is shown in Figure \ref{fig:decomp_img} without the presence of noise $\ve$.

\begin{figure}[hbt]
    \centering
    \includegraphics[width=0.7\textwidth]{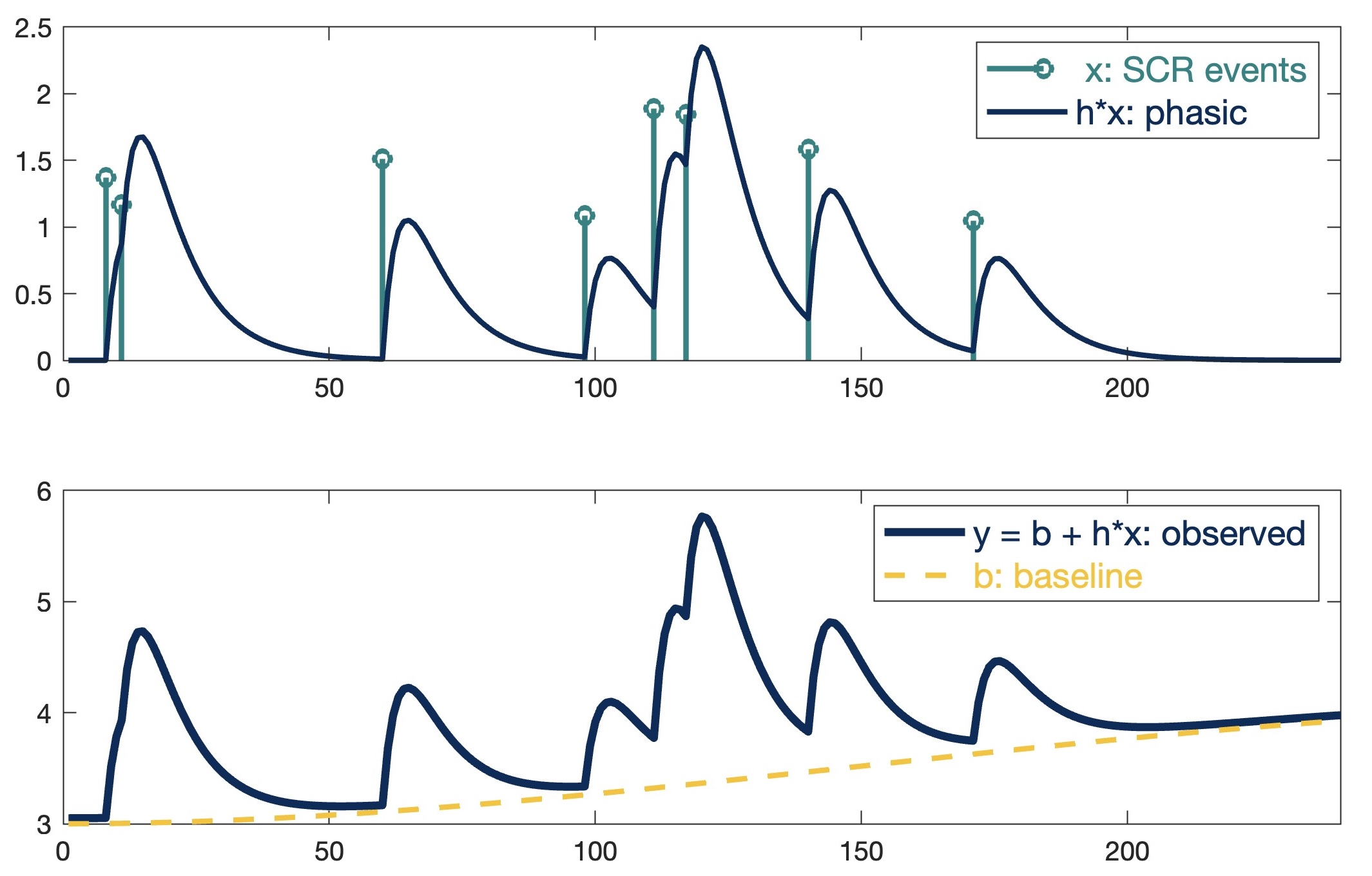}
    \caption{The decomposition of an EDA signal into the baseline and signal events.}
    \label{fig:decomp_img}
\end{figure}

Given the notations above, our  goal of EDA decomposition problem is an inverse problem of  recovering $\vx$ from $\vy$ in \eqref{equ:model} where $\vh$ is known. 
This setup has been widely adopted in the literature such as~\cite{A05, eda_cs, spEDA}. 
Ledalab~\cite{BK10} is an early Matlab package on EDA signal processing, but known to be sensitive to noise and artifacts.
Neurokit2~\cite{neurokit2} is a popular Python toolbox that provides a comprehensive suite of processing routines for a variety of bodily signals including EDA.
Leveraging the success in compressed sensing, recent work has been exploiting sparsity for improved accuracy and computational efficiency.
The work~\cite{chaspari2014sparse}  modeled the SCR signal as a sparse linear combination of atoms of a dictionary. 
Jain et al \cite{eda_cs} proposed a new compressed sensing (CS) framework~\cite{C08, CDD09, CWW14} that is more robust to motion artifact for recovering the SCR components. This work will be discussed in more detail in Section \ref{sec:cs}.
Wavelet transformation was used in \cite{shukla2018efficient} with demonstrated efficiency.
Hernando-Gallego et al. \cite{spEDA} developed sparsEDA, a nonnegative sparse deconvolution method where SCR and baseline are jointly recovered.


While these methods are effective in many scenarios, they typically process signals in isolation. By failing to exploit the shared structure across multiple signals and multiple subjects, these approaches remain sensitive to baseline fluctuations and high noise levels. Such limitations demonstrate the need for more robust frameworks that incorporate additional structural priors beyond simple sparsity.

\subsection{Contributions and Organization}

In this work, we introduce \textbf{gmsEDA}, a novel framework based on Generalized Matrix Separation (GMS)~\cite{CW25} designed to process multiple EDA signals jointly rather than in isolation. The primary contributions of this paper are three-fold: 
\begin{itemize}
\item We move away from traditional single signal processing by stacking concurrent or segmented recordings into a single data matrix. This formulation allows the algorithm to exploit collective low-rank baseline structures and cross-signal event sparsity, significantly enhancing the recovery of underlying physiological components. The modeling can handle noise and motion artifacts better. Our method also has the flexibility to either analyze data on an individual-subject basis in multiple segments or process signals from multiple subjects jointly.

\item Our method is founded on matrix separation theory and theoretical guarantees are provided in Appendix \ref{sec:theory}. 

\item We perform extensive numerical validations on both synthetic benchmarks and real-world datasets. In synthetic trials spanning four distinct signal models, gmsEDA demonstrates superior accuracy and robustness compared to the CS method~\cite{eda_cs}. Furthermore, when validated on real-world physiological data collected from a 23-subject affective study, gmsEDA  outperforms standard CS methods, sparsEDA, and the widely used NeuroKit2 package by achieving higher event match rates and lower false peak rates. 
    
\end{itemize}

The remainder of this paper is organized as follows. Section \ref{sec:sig} outlines the mathematical formulation of the EDA signal models, as well as a review of the CS method. Section \ref{sec:gms} details the proposed gmsEDA framework and its corresponding overlapped reshape procedure. Numerical experiments on synthetic and empirical data are presented in Section \ref{sec:exp}, and Section \ref{sec:dis} concludes the paper with a discussion of implications and future research directions.  Appendix \ref{sec:math} provides mathematical background and is recommended to be read with Section \ref{sec:sig}-\ref{sec:gms}. Appendix \ref{sec:theory} provides theoretical support for gmsEDA and Appendix \ref{sec:exp_detail} details the real data experiment.

\section{Signal Models}\label{sec:sig}


Throughout this paper, we use boldface letters such as $\vx, \vy$ to represent vectors, boldface uppercase letters such as $\mX, \mY$ to represent matrices.

For a vector $\vv=(v_1, v_2, \cdots, v_n)\in\R^n$,
its $\ell_p$ norm is $\|\vv\|_p=(\sum_{i=1}^n|v_i|^p)^{1/p}$ for any $p\geq1$. Its infinity norm is $\|\vv\|_\infty=\max_{i=1}^\infty|v_i|$.
We also use the notation $\|\vv\|_0$ to denote the number of nonzero entries in $\vv$. A vector $\vv$ is called \emph{$s$-sparse} if $\|\vv\|_0\leq s$.
We let $\vv_s$ be the $s$-sparse vector that keeps the $s$ largest coordinates (in magnitude) of $\vv$ while setting all other coordinates to 0.
For a matrix $\mA=(a_{ij})\in\R^{m\times n}$, its nuclear norm, denoted by $\|\mA\|_*$, is the sum of all singular values (see Appendix \ref{sec:math}). The matrix $\ell_1$ norm is $\|\mA\|_1=\sum_{i,j}|a_{ij}|$.
We also use $[n]$ for the index set$\{1,2,\cdots,n\}$.

\subsection{SCR events and Baseline Signals}
The backward difference matrix $\mD$ is
\begin{equation}\label{equ:D}
\mD=\begin{bmatrix}
1&-1&0&\cdots&0\\
0&1&-1&\cdots&0\\
\vdots &\vdots&\ddots&\ddots&\vdots\\
0&0&\cdots&1&-1
\end{bmatrix}\in\R^{(n-1)\times n}.
\end{equation}

The SCR events signal $\vx\in\R^n$ is sparse and we let $s$ denote its maximal number of nonzero entries. 
In practice, $\vx$ is approximately sparse, which will be quantified by  $\delta>0$. The SCR events signal $\vx$ lies in the following set
\begin{equation}\label{equ:Xsd}
\Xc_{s,\delta}:=\{\vx\in\R^n:\|\vx-\vx_s\|_1\leq\delta\}.
\end{equation}

For the baseline signals, we will have two different models. The first model is a more traditional model where the baseline vector $\vb$ is slowly varying. Let 
\begin{equation}
\Bc:=\{\vb\in\R^n: \|\mD\vb\|_\infty\leq \frac{1}{n}\}.
\end{equation}
be the set of slowly varying signals. To allow for noise at the level $\gamma>0$, we define
\begin{equation}\label{equ:B1g}
\Bc_\gamma:=\{\vb\in\R^n: \mD\vb=\mD\vf+\ve,\text{ for some }\vf\in\mathcal{BC}, \|\ve\|_1\leq\gamma\}.
\end{equation}
Signals in $\Bc_{\gamma}$ only differ from signals in $\Bc$ by a small perturbation $\ve$ in terms of the backward difference. This is referred as the BC model\footnote{``B'' stands for baseline and ``C'' stands for continuously/slowly varying.}.

The second model follows the baseline model in \cite{eda_cs} to account for discontinuous jumps, which models the motion artifact in wearable devices. Let
\begin{equation}
\Bj_{k}:=\{\vb\in\R^n: \|\mD\vb\|_0\leq k\}.
\end{equation}
be the set of  signals whose difference is $k$-sparse. Similar to \eqref{equ:B1g}, we define a perturbation of $\Bj_k$ at the level $\gamma$:
\begin{equation}
\Bj_{k,\gamma}:=\{\vb\in\R^n: \mD\vb=\mD\vx+\ve,\text{ for some }\vx\in\Bj_{k}, \|\ve\|_1\leq\gamma\}.
\end{equation}
We call $\Bj_{k,\gamma}$ the BJ model.

With $n=360, \gamma=1$, Figure \ref{fig:b}(a) shows a signal in  baseline model  $\Bc_\gamma$  and Figure \ref{fig:b}(b) shows a signal in  baseline model  $\Bj_{k,\gamma}$ with $k=3$ jumps.
\begin{figure}[htb]
\includegraphics[width=\textwidth]{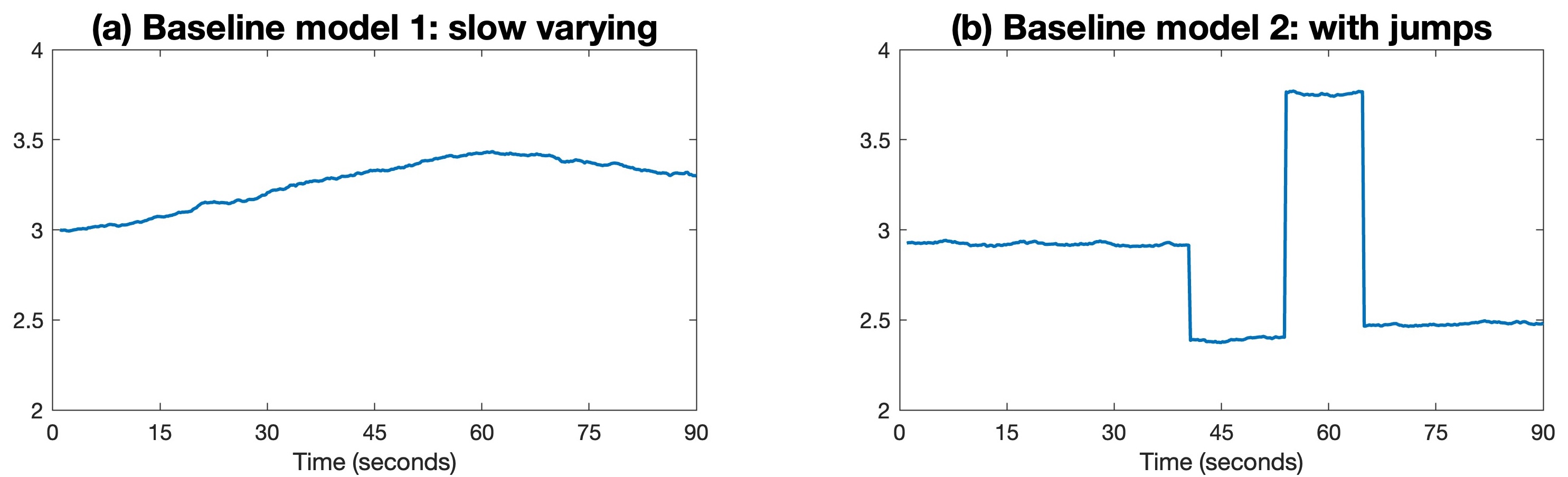}
\caption{Illustration of different baseline models. (a) Baseline model 1: the BC model where $b\in\Bc_\gamma$; (b) Baseline model 2: the BJ model where $b\in\Bj_{k,\gamma}$}\label{fig:b}
\end{figure}

\subsection{The filter $\vh$}
Given $\tau_1>\tau_2>0$, we define the kernel function
\begin{equation}\label{equ:h}
f(t)=2(e^{-t / \tau_1} - e^{-t / \tau_2}), \quad t\geq0,
\end{equation}
which is the bi-exponential impulse response from the psychophysiology literature~\cite{A05}.
Figure~\ref{fig:f} displays this function for a particular choice of the parameters $\tau_1, \tau_2$.

\begin{figure}[htb]
\centering
\begin{tikzpicture}[scale=2.2]
\draw[->](0,0)--(2.3,0) node[below]{$t$};
\draw[->](0,0)--(0, 1.6);
\filldraw (0,1) circle(0.5pt) node[left]{1};
\filldraw (0,1.5) circle(0.5pt) node[left]{1.5};
\filldraw (0,0.5) circle(0.5pt) node[left]{0.5};
\filldraw (1,0) circle(0.5pt) node[below]{20};
\filldraw (2,0) circle(0.5pt) node[below]{40};
\filldraw (0.5,0) circle(0.5pt) node[below]{10};
\filldraw (1.5,0) circle(0.5pt) node[below]{30};
\draw[ultra thick] plot[domain=0:2, samples=100, smooth] (\x, {2*exp(-\x/10*20)-2*exp(-\x*20)});
\draw (1, -0.3) node{Time (seconds)};
\end{tikzpicture}
\caption{Impulse Response $f$ when $\tau_1=10, \tau_2=1$}\label{fig:f}
\end{figure}
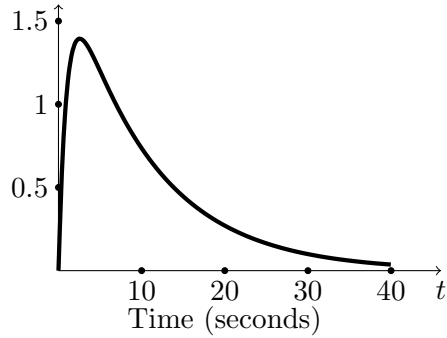

The boldface $\vh$ will be a discretized version of $f$, sampled at 4Hz in the interval $t\in[0, T]$, where $T$ is the duration of the signal in seconds.
We will also pad 0's to the end of $h$ so that $h$ has the same length as each signal.
Therefore $\vh=(f(0.25), f(0.5), \cdots, f(T), 0, \cdots, 0)\in\R^{n}$.
Specifically, $\vh_j$, the $j$th coordinate of $\vh$, is the following
$$\vh_j=\left\{\begin{array}{ll}
f(0.25(j-1)), &j=1, 2, \cdots, 4T\\
0, &j = 4T+1, \cdots, n\end{array}\right.$$


The convolution $\vh*\vx$ can be written as a matrix vector multiplication $\mH\vx$ where the matrix $\mH$ is the convolution matrix associated with $\vh$, so \eqref{equ:model} can also be written as
\begin{equation}\label{equ:modelH}
\vy=\vb+\mH\vx+\ve.
\end{equation}
In our modeling, we use a square $\mH$ as
\begin{equation}\label{equ:H}
\mH=\begin{bmatrix}
\vh_1&0&\cdots&0\\
\vh_2& \vh_1&\vdots&\vdots\\
\vdots &\vdots&\ddots&0\\
\vh_n&\vh_{n-1}&\vdots&\vh_1
\end{bmatrix}\in\R^{n\times n},
\end{equation}
which is the top square submatrix of the convolution matrix in \cite{eda_cs}.

\subsection{Review of the Compressed Sensing Based Decomposition}\label{sec:cs}
A compressed sensing based approach was proposed in \cite{eda_cs}.
Let $\vy_0=\vb+\mH\vx+\ve$ be the observed EDA signal. The difference operator $\mD$ is then applied to the equation to reduce the motion artifact of the baseline. 
 This gives the equation
\begin{equation}
\mD\vy_0= \mD\vb+\mD\mH\vx+\mD\ve,
\end{equation}
which can be rewritten as 
\begin{equation}\label{equ:csb}
\mD\vy_0=\begin{bmatrix}\mD\mH&\mI\end{bmatrix}\begin{bmatrix}\vx\\ \mD\vb\end{bmatrix}+\mD\ve.
\end{equation}
For the BJ model where $\vb\in\Bj_{k,\gamma}$, $\mD\vb$ is approximately $k$ sparse, and therefore all together $\vz_0=\begin{bmatrix}\vx\\ \mD\vb\end{bmatrix}$ is approximately sparse.

Let $\mA=\begin{bmatrix}\mD\mH&\mI\end{bmatrix}$, then \eqref{equ:csb} is 
\begin{equation}
\mD\vy_0 = \mA\vz_0+\mD\ve.
\end{equation}

The recovery of $\vz_0$ can then be formulated as the Least Absolute Shrinkage and Selection Operator (Lasso)~\cite{hastie, fista}:
\begin{equation}\label{equ:lasso}
\hat \vz=\argmin_{\vz\in\R^{2n-1}}\{\lambda\|\vz\|_1+\frac{1}{2}\|\mD\vy_0-\mA\vz\|_2^2\},
\end{equation}
with an appropriately chosen $\lambda$.  
The recovered SCR is the first $n$ coordinates of $\hat z$ as
\begin{equation}\label{equ:csx}
\hat \vx = \hat \vz_{[n]}.
\end{equation}
We use $\hat \vx=CS(\vy_0)$ to indicate the CS method which combines \eqref{equ:lasso} and \eqref{equ:csx}. 

\section{Our Decomposition Method using Matrix Separation}\label{sec:gms}
The key idea of our method is to recover from multiple EDA signals jointly using matrix separation.
Let $\vy_1, \vy_2, \cdots, \vy_K\in\R^n$ be $K$ observed EDA signals, and we have 
\begin{equation}\label{equ:gms1}
\vy_i=\vb_i+\mH\vx_i+\ve_i, \quad i=1, 2, \cdots, K
\end{equation} following \eqref{equ:modelH}. 
These $K$ signals could come from one subject or multiple subjects.
Our model will process all $K$ signals simultaneously as we put them together as columns of one matrix. Let $\mY_0=[\vy_1,\cdots,\vy_K], \mB_0=[\vb_1, \cdots, \vb_K], \mX_0=[\vx_1,\cdots,\vx_K], \mE=[\ve_1,\cdots,\ve_K]$, \eqref{equ:gms1} then becomes
\begin{equation}\label{equ:gms2}
\mY_0=\mB_0+\mH\mX_0+\mE.
\end{equation}
The matrix $\mB_0$ will be an approximately low rank matrix given the baseline signal is slowly varying (with or without a few jumps). See Section \ref{sec:syn_i} for a simulated example. The matrix $\mX_0$ will be a sparse matrix since every column is sparse. 

Our decomposition problems becomes recovering $\mX_0$ given $\mY_0$ in \eqref{equ:gms2}, where $\mH$ is known and the noise $\mE$ is reasonably controlled. Theoretically, such a problem was first raised and explored in \cite{CW25} by solving the convex optimization problem
\begin{equation}\label{equ:gms}
(\hat \mB,  \hat \mX)=\argmin_{\mB, \mX}\{\|\mB\|_*+\lambda \|\mX\|_1\} \quad \text{subject to }\mY_0=\mB+\mH\mX.
\end{equation}
The alternating direction method of multipliers (ADMM)~\cite{boyd2011distributed} can be used to solve \eqref{equ:gms}. 
While equation \eqref{equ:gms2} includes an explicit noise term $\mE$, the strict equality constraint in \eqref{equ:gms} can be enforced during optimization because ADMM naturally handles small perturbations, distributing the residual noise into the recovered sparse and low-rank components. 

We will solve a relevant, but not necessarily equivalent problem that was proposed in \cite{CD25}. Let $\mU_\mH\mathbf{\Sigma}_\mH\mV_\mH^\top $ be a reduced SVD (see Appendix \ref{sec:math}) of $\mH$ and let $\mC=\mU_\mH\Sigma_\mH^{-1}\mU_\mH^\top$. It is argued in \cite{CD25} that $(\hat \mB_c,  \hat \mX_c)$ from the following program is more accurate, robust, and computationally efficient at recovering $(\hat \mB_0,  \hat \mX_0)$.
\begin{equation}\label{equ:gmsc}
\begin{cases}
(\hat{\mW}_c, \hat{\mX}_c) = \displaystyle \argmin_{\mW, \mX} \{\|\mW\|_*+\lambda \|\mX\|_1\} , \quad \text{subject to }  \mC \mY_0=\mW + \mC\mH\mX  \\
\hat{\mB}_c = \mY_0 - \mH \hat{\mX}_c.
\end{cases}
\end{equation}
Intuitively, \eqref{equ:gmsc} is obtained by  multiplying the constraint equation in \eqref{equ:gms} by $\mC$ on the left.
The new filter is $\tilde \mH =\mC\mH=\mU_\mH\mV_\mH^\top$ whose condition number is 1. Theoretical guarantee of using \eqref{equ:gmsc} for recovering $\mX_0$ (and $\mB_0$ if desired) is presented in Appendix \ref{sec:theory}.

The framework of general matrix separation is broader than decomposing EDA  signals.  For example, In the setup \eqref{equ:gms2}, the sparse matrix $\mX_0$ can have negative coordinates and still be successfully recovered via \eqref{equ:gms}.
The work \cite{CD25} talked about other applications such as simultaneous video background separation and deblurring.

%
%

\subsection{Overlapped Reshape}\label{sec:or}
The GMS method we proposed seemingly requires \textbf{multiple} EDA signals. What if we only wish to process one single EDA signal? This section addresses this question by performing a simple signal reshape.

We first pick the number of cuts $C\geq2$ and the overlap ratio $q\in[0,1)$. Then a single EDA signal $\vy=(y_1, y_2, \cdots, y_N)$ can be transformed into a  matrix $\OR_{C,q}(\vy)=[\vw_1, \cdots, \vw_m]$ following the steps below:
\begin{itemize}
\item The length of each $\vw_i$ is $n=\lfloor N/C\rfloor$. The overlapped length is $t=\lfloor nq\rfloor$.
\item $m=\left\{\begin{array}{ll}
\lfloor\frac{N-n}{n-t}\rfloor+2,& \text{if }\lfloor\frac{N-n}{n-t}\rfloor<\frac{N-n}{n-t}\\
\lfloor\frac{N-n}{n-t}\rfloor+1, & \text{if }\lfloor\frac{N-n}{n-t}\rfloor=\frac{N-n}{n-t}
\end{array}\right.$
\item $\vw_i=[y_{(i-1)(n-t)+1}, \cdots, y_{(i-1)(n-t)+n}]^\top$ for $i=1, 2, \cdots, m-1$. This ensures the later $q$ portion of $\vw_i$ is the same as the first $q$ portion of $\vw_{i+1}$.
\item $\vw_m=[y_{N-n+1}, \cdots, y_N]^\top$.
\end{itemize}

This transformation can be thought of as an \emph{overlapped reshape}, hence the function name $\OR$. It is as simple as rearranging the EDA signal $\vy$ into an appropriately sized matrix column-wise if we assume each column has no overlap. For example when $N=12, C=2, q=0$, we have
$$
 \OR_{2,0}\left(\begin{bmatrix}y_1\\y_2\\\vdots\\y_{12}\end{bmatrix}\right)=\begin{bmatrix}
y_1&y_7\\
y_2&y_8\\
\vdots&\vdots\\
y_6&y_{12}
\end{bmatrix}.
$$

As another illustrative example, let $N=12, C=2, q=0.5$, then we have the following transformation/reshape:
$$
 \OR_{2,0.5}\left(\begin{bmatrix}y_1\\y_2\\\vdots\\y_{12}\end{bmatrix}\right)=\begin{bmatrix}
y_1&y_4&y_7\\
y_2&y_5&y_8\\
\vdots&\vdots&\vdots\\
y_6&y_9&y_{12}
\end{bmatrix}.
$$
The parameter $C=2$ determines that each new EDA signal length is $N/2=6$. Since $q=0.5$, we then have an overlap of 50\% between neighboring columns.

Using the reshaped matrix $\OR_{C,q}(\vy)\in\R^{n\times m}$ as the input $\mY_0$ in \eqref{equ:gms}, the recovered SCR events $G(\OR_{C,q}(\vy))$ needs to be transformed back to $N\times 1$. This inverse transformation, denoted by $\OR_{C,q}^{-1}$, is simply a special vectorization, where the average value is used for the overlapped portion. As an example, we use $N=12, C=2, q=0.5$ again, then
$$
\OR_{2,0.5}^{-1}\left(\begin{bmatrix}x_{11}&x_{12}&x_{13}\\
x_{21}&x_{22}&x_{23}\\
x_{31}&x_{32}&x_{33}\\
x_{41}&x_{42}&x_{43}\\
x_{51}&x_{52}&x_{53}\\
x_{61}&x_{62}&x_{63}
\end{bmatrix}\right)=\begin{bmatrix}
x_{11}\\
x_{21}\\
x_{31}\\
0.5x_{41}+0.5x_{12}\\
0.5x_{51}+0.5x_{22}\\
0.5x_{61}+0.5x_{32}\\
0.5x_{42}+0.5x_{13}\\
0.5x_{52}+0.5x_{23}\\
0.5x_{62}+0.5x_{33}\\
x_{43}\\
x_{53}\\
x_{63}
\end{bmatrix}.
$$
In particular, $\OR_{C,q}^{-1}(\OR_{C,q}(\vy))=\vy$.

If the input is a matrix $\mY=[\vy_1, \cdots, \vy_K]$, then $\OR_{C,q}(\cdot)$ process it column-wise as 
$$\OR_{C,q}(\mY) = [\OR_{C,q}(\vy_1), \cdots, \OR_{C,q}(\vy_K)].$$
The inverse transformation $\OR_{C,q}^{-1}(\cdot)$ is processed similarly if the input is from reshaping a matrix consisting of multiple EDA signals, as long as we keep track of the shape of $\OR_{C,q}(\vy_i)$.

There are several advantages with this overlapped reshaping preprocess.
\begin{itemize}
\item This allows our GMS method to be applied to a single EDA signal.
\item If we simply chop the EDA signals into non-overlapping pieces, the detection of peaks may be missed if they are near the cutoff. Overlapped reshape circumvents this issue. Figure \ref{fig:eda_chop} shows the reshape of $\vy\in\R^{240}$ from Figure \ref{fig:decomp_img} with the parameters $C=2, q=0.5$. A peak occurs near index 120 which is the end of the first piece, but has a chance to be recovered from the second piece.
\item This overlapped reshape can be applied to any other decomposition method such as in \eqref{equ:csp}. This increases computational efficiency given the decreased signal length. Related numerical experiments are conducted in Section \ref{sec:syn_p}.
\end{itemize}


\begin{figure}[htbp]
\centering
\includegraphics[width=0.8\textwidth]{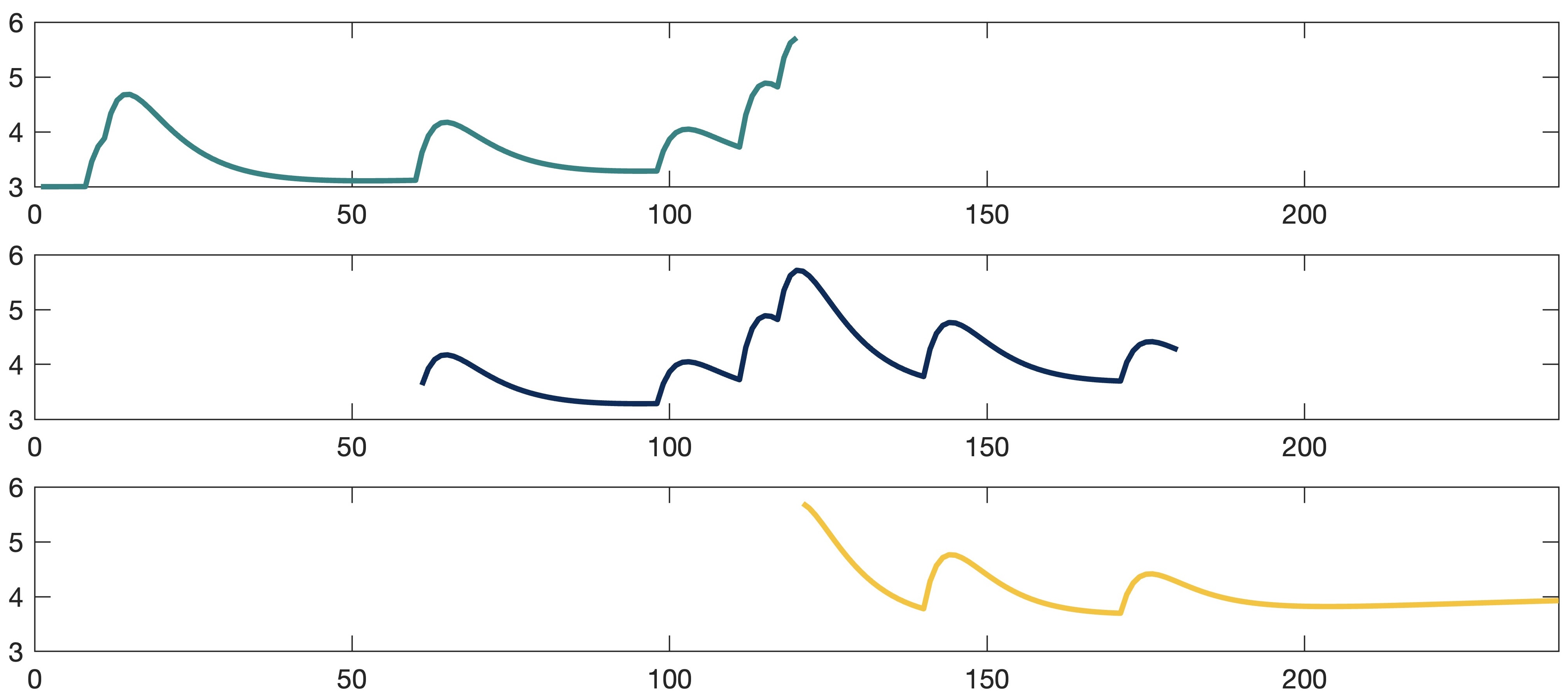}
\caption{The simulated EDA signal from Figure \ref{fig:decomp_img} has been reshaped into three pieces on the same timestamp, with 50\% overlap. These 3 pieces will then become the 3 columns of the reshaped matrix, with aligned index.}
\label{fig:eda_chop}
\end{figure}

\subsection{Model Summary}
Given $\mY_0=[\vy_1,\cdots, \vy_K]\in\R^{n\times K}$, 
\begin{sialsbox}{GMS and GMS-P}{}
We use $\hat \mX_c=G(\mY_0)$ from \eqref{equ:gmsc} to indicate the GMS method.  If we use overlapped reshape, we will call it GMS-P where P stands for parallel.
\begin{equation}\label{equ:gmsp}
\text{GMS-P:}\quad\begin{array}{ll}
1. & \mR_{\mY}=\OR_{C, q}(\mY_0)\\
2. & \mR_{\mX} = G(\mR_{\mY})\\
3. & \hat \mX = \OR_{C, q}^{-1}(\mR_{\mX})
\end{array}
\end{equation}

\end{sialsbox}

As mentioned, the overlapped reshape can also be applied to other methods. For example, given observed $\mY_0$, the 3 steps of CS-P is 
\begin{equation}\label{equ:csp}
\text{CS-P:}\quad\begin{array}{ll}
1. & \mR_{\mY}=\OR_{C, q}(\mY_0)\\
2. & \mR_{\mX} = CS(\mR_{\mY}) \text{ (see Section \ref{sec:cs})}\\
3. & \hat \mX = \OR_{C, q}^{-1}(\mR_{\mX})
\end{array}
\end{equation}
However, unlike $G(\cdot)$, the function $CS(\cdot)$ still processes the whole matrix column-by-column, not utilizing any interaction between the columns. Therefore it is expected that CS and CS-P will have similar recovery performance, but CS-P may gain computational efficiency. This is demonstrated in Section \ref{sec:syn_p} (Table \ref{tab:pr} in particular).


\subsubsection*{Post-processing}
As mentioned, our GMS method can recover the sparse matrix $\mX$ even when the entries are negative. However, for the EDA signal decomposition problem, it is usually assumed that each phasic component is positive. For both GMS and CS method, we simply set all the negative entries of $\hat \mX$ to 0.

\section{Numerical Experiments}\label{sec:exp}
We conduct extensive numerical experiments on both synthetic and real world data. The experiments were run on a MacBook Pro with Apple M3 chip and 8GB RAM, using Matlab 2023b or Python 3.11.14. Python was only used when using the package NeuroKit2.

\subsection{Experiments On Synthetic Data}\label{sec:syn}

To generate the SCR events signal $\vx\in\Xc_{s,\delta}$ (as defined in \eqref{equ:Xsd}), we first  pick $s$ coordinates uniformly at random as the  support of $\vx$. The value of each nonzero coordinate are i.i.d., following one of the two different random distributions: exponential distribution with mean 2 (exponential model, labeled as XE) or uniform distribution in the interval $[2,7]$ (uniform model, labeled as XU). Finally, a standard Gaussian vector, rescaled to have $\ell_1$ norm $\delta$, is added to the sparse vector. 
We consider the XE model more challenging since some of its coordinates can be  close to 0. 

As described in Section \ref{sec:sig}, there are two models for the baseline signal $\vb$. Recall that the continuous and slow varying model is labeled BC and BJ is the model with jumps. 

The noise vector $\ve$ follows the normal distribution and then rescaled such that the $\ell_2$ norm is $\epsilon$. This is summarized in Table \ref{tab:random}.

\begin{table}[htbp]
\caption{Details of Signal Models with Random Noise}\label{tab:random}
\begin{tabular}{c | c | c | c | c}
Name&\multicolumn{4}{c}{Description}\\
\hline
XE&\multirow{2}{*}{SCR $\vx$} &\multirow{2}{*}{$\vx=\vx_s+\ve_\vx\in\Xc_{s,\delta}$}& entries of $\vx_s\stackrel{i.i.d.}{\sim}\text{exp}(2)$ &\multirow{2}{*}{$\|\ve_\vx\|_1=\delta$}\\\cline{1-1}\cline{4-4}
XU&&&entries of $\vx_s\stackrel{i.i.d.}{\sim}\text{unif}([2,7])$\\
\hline
BC&\multirow{2}{*}{baseline $\vb$}&$\vb=\vf+\ve_\vb\in\Bc_\gamma$&$\vf\in\Bc$ generated by spline fitting&\multirow{2}{*}{$\|\mD\ve_\vb\|_1=\gamma$}\\\cline{1-1}\cline{3-4}
BJ&&$\vb=\vf+\ve_\vb\in\Bj_{k,\gamma}$& $\vf\in\Bj_k$, entries of $\mD\vf\stackrel{i.i.d.}{\sim}\text{N}(0,1)$\\
\hline
&noise $\ve$ & \multicolumn{3}{c}{coordinates of $\ve\stackrel{i.i.d.}{\sim}\text{N}(0,1)$, followed by scaling such that $\|\ve\|_2=\epsilon$} \\
\hline
\end{tabular}
\end{table}

The filter $\mH$, as defined in \eqref{equ:H}, solely depends on the values of $\tau_1$ and $\tau_2$. We use $\tau_1=2, \tau_2=0.75$ \cite{BK10} in all the synthetic experiments except for Section \ref{sec:tau12}.

The EDA signal $\vy$ is then generated using formula \eqref{equ:modelH}. We say that $\vy$ follows the XE-BC model if the SCR events signal $\vx$ follows the XE model and the baseline $\vb$ follows the BC model. The same goes for the other three combinations. The work \cite{eda_cs} uses the XE-BJ model.

For synthetic experiments, we only compare our method to the CS method due to similarity in signal models.

After some tuning, we use a universal $\lambda=0.02$ in \eqref{equ:lasso} for the CS method and $\lambda=\frac{3}{\sqrt{\max(\text{size}(\mY_0,1), \text{size}(\mY_0,2))}}$ in \eqref{equ:gmsc} for the GMS method. \textbf{These parameter choices are used in Section \ref{sec:real} as well}.

\subsubsection{Initial experiments}\label{sec:syn_i}

The first experiment has the noise level $\epsilon=0.3$ and sparsity $s=10$ fixed.
We randomly generates SCR events signal $\vx\in\R^{370}$ with $\delta=10$, baseline $\vb$ with $\gamma=10$. $K=40$ instances were created, and the average recovery relative error $\frac{\|\hat \vx-\vx\|_2}{\|\vx\|_2}$ over these 40 instances are calculated. Both SCR models XE and XU, and both baseline model BC and BJ are tested, which resulted  Table \ref{tab:rel}. 

\begin{table}[htb]
\caption{Mean Relative Error By Both Methods: better results are in boldface}\label{tab:rel}
\centering
\begin{tabular}{cc|cc|cc|cc}
&&\multicolumn{4}{c}{Baseline $\vb$ model}\\
&&\multicolumn{2}{c}{BC}&\multicolumn{2}{c}{BJ, 1 jump} &\multicolumn{2}{c}{BJ, 2 jumps}\\
\hline
&&CS &GMS &CS &GMS &CS &GMS\\
\multirow{2}{*}{SCR $\vx$ model}&XE&0.385 & \textbf{0.203 }&0.457 & \textbf{0.298} & 0.465 & \textbf{0.460}\\
&XU&0.406 & \textbf{0.132}&0.359 & \textbf{0.169} & 0.371 & \textbf{0.230}\\
\end{tabular}
\end{table}

We see from Table \ref{tab:rel} that our GMS method performs better for all 6 signal models. The CS method was motivated to recover events more robustly against model motion artifacts (BJ model). In this initial experiment, we demonstrate that our GMS method works even better in this scenario.

We also test that the simulated baseline matrix $\mB\in\R^{370\times40}$ from the BJ model (2 jumps) is approximately rank-1 as $\frac{\|\text{best rank-1 appr. of }\mB_0\|_F}{\|\mB_0\|_F}=0.987$.

Figure \ref{fig:initial} plots one instance (out of 40) of the recovered SCR events $\hat \vx$ against the ground truth, for both methods when the models XU and BJ are used. Both methods detect the stimulus events location quite effectively, but are both underestimating the magnitudes of the strength, with our GMS method being more accurate. We also framed two locations where events detection errors are made by the CS method.
\begin{figure}[htbp]
\centering
\includegraphics[width=0.85\textwidth]{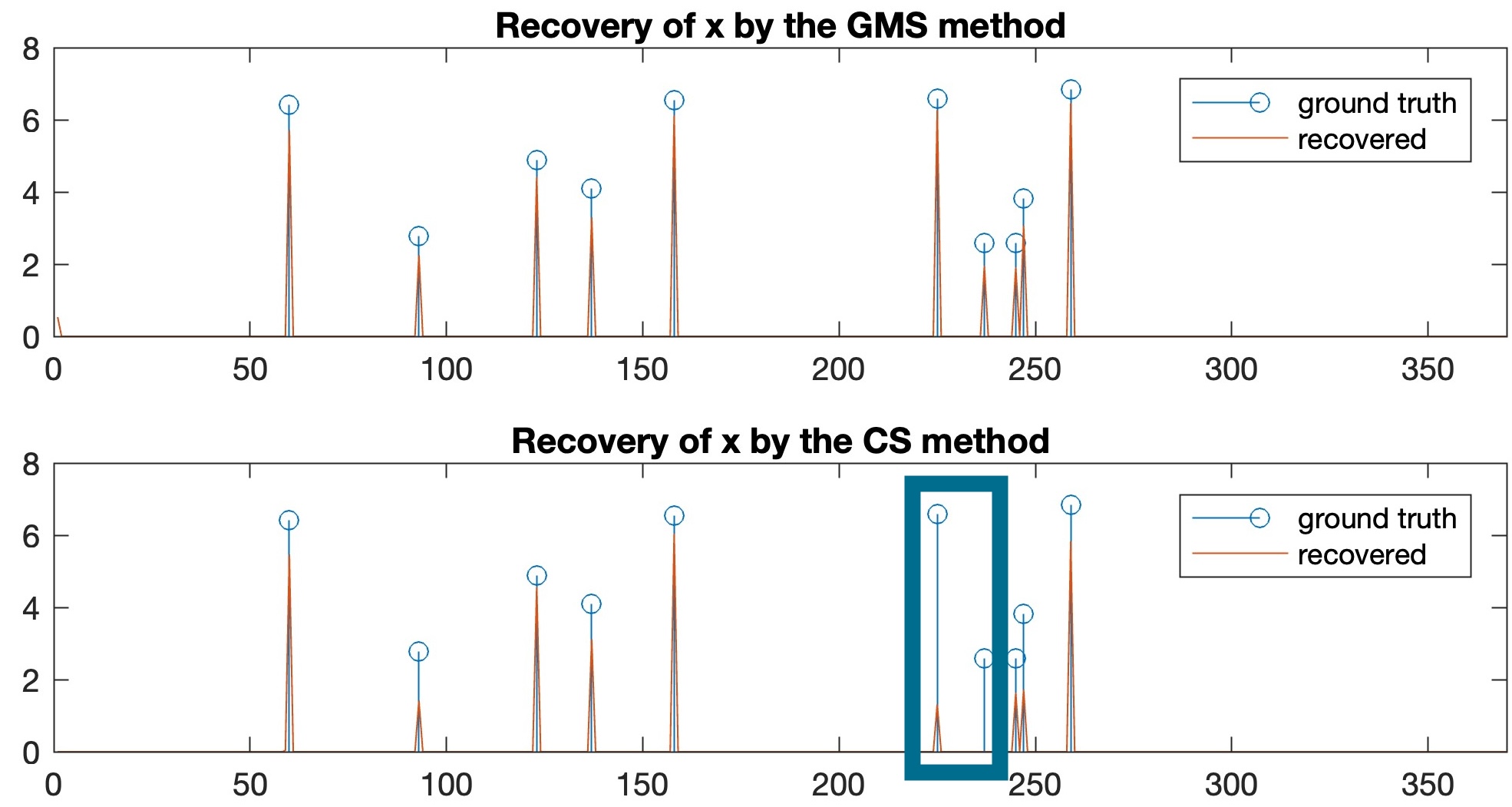}
\caption{Recovered $\vx$ for the XU-BJ (1 jump) model. Other parameters are $s=10, \delta=10, \gamma=10, \epsilon=0.3$. In the frame, one event is detected but with minimal magnitude while the other event is missed entirely by the CS method.}\label{fig:initial}
\end{figure}

\subsubsection{Exploration on sparsity level and noise level}
 In this experiment, we test on a range of noise and sparsity levels. Specifically, the noise level $\epsilon$ will be one of these six values: $\{0.02, 0.08, 0.16, 0.32, 0.64, 1.28\}$, and the sparsity $s$ will be  from the list $\{1,4,7,10,\cdots, 28,31\}$. We let $n=240, K=40, \delta=1, \gamma=1$ be fixed.

For each fixed $\epsilon$ and $s$, we randomly generate $K=40$ instances of $\vx, \vb, \ve$ (and therefore $\vy$) following Table \ref{tab:random}. We then compute the mean relative error on recovering $\vx$, averaged over these 40 trials by either the CS method or the GMS method. Figures \ref{fig:heat1}-\ref{fig:heat4} show heatmaps of mean relative error for the four different signal models listed in Table \ref{tab:random}.

Figure \ref{fig:heat1} shows heatmaps when $\vy$ follows XE-BC model. The GMS method produces much better results across the board. Moreover, the GMS method allows for bigger sparsity and noise level.

\begin{figure}[htb]
\includegraphics[width=\textwidth]{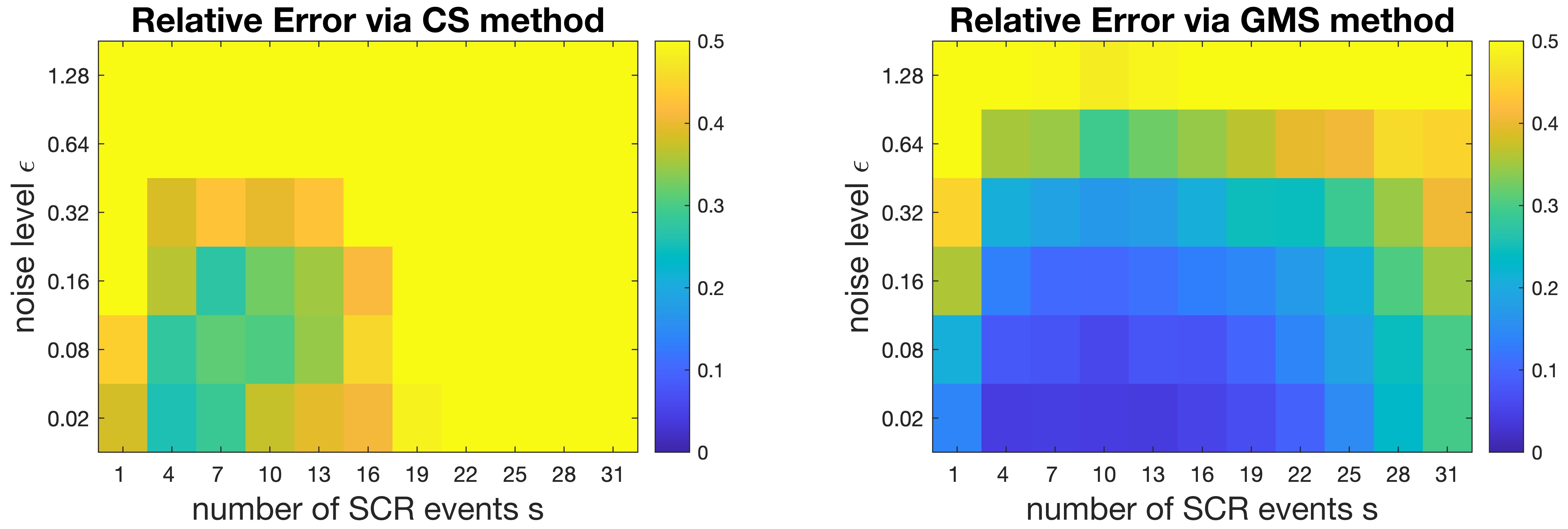}
\caption{Mean relative error for various $s$ and $\epsilon$ with $\vx$ and $\vb$ follow the \textbf{XE-BC} combination}\label{fig:heat1}
\end{figure}

Figure \ref{fig:heat_mu} shows heatmaps when $\vy$ follows the XU-BC model. Once again, the GMS method produces much better results across the board. For the uniform signal model, the recovery performance has a sharp cutoff on certain sparsity level. For the CS method, recovery is poor if $s\geq13$ and for the GMS method, the recovery is poor if $s\geq19$.

\begin{figure}[htb]
\includegraphics[width=\textwidth]{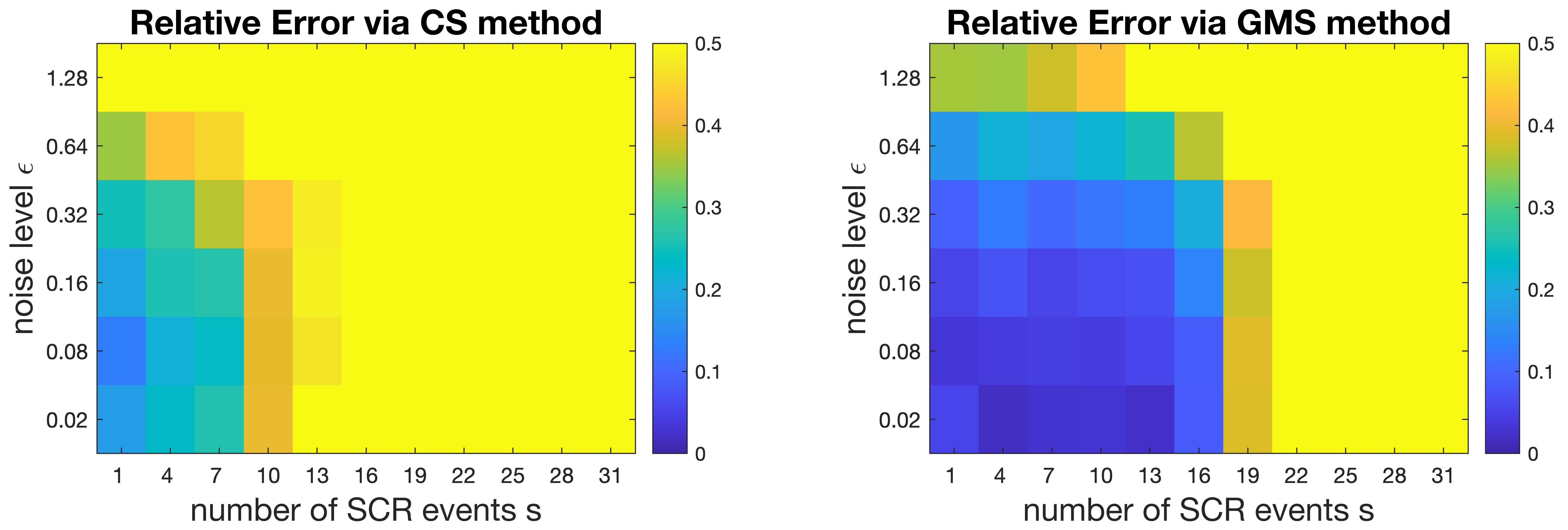}
\caption{Mean relative error for various $s$ and $\epsilon$ with $\vx$ and $\vb$ follow the \textbf{XU-BC} combination}\label{fig:heat_mu}
\end{figure}

Figure \ref{fig:heat2} shows heatmaps when $\vy$ follows the XU-BJ model with 1 random shift. The CS method produces better recovery for small sparsity values $s=1,4,7$. However, our GMS method performs better with bigger sparsity even though the baseline model BJ is more favorable to the CS method.

\begin{figure}[htb]
\includegraphics[width=\textwidth]{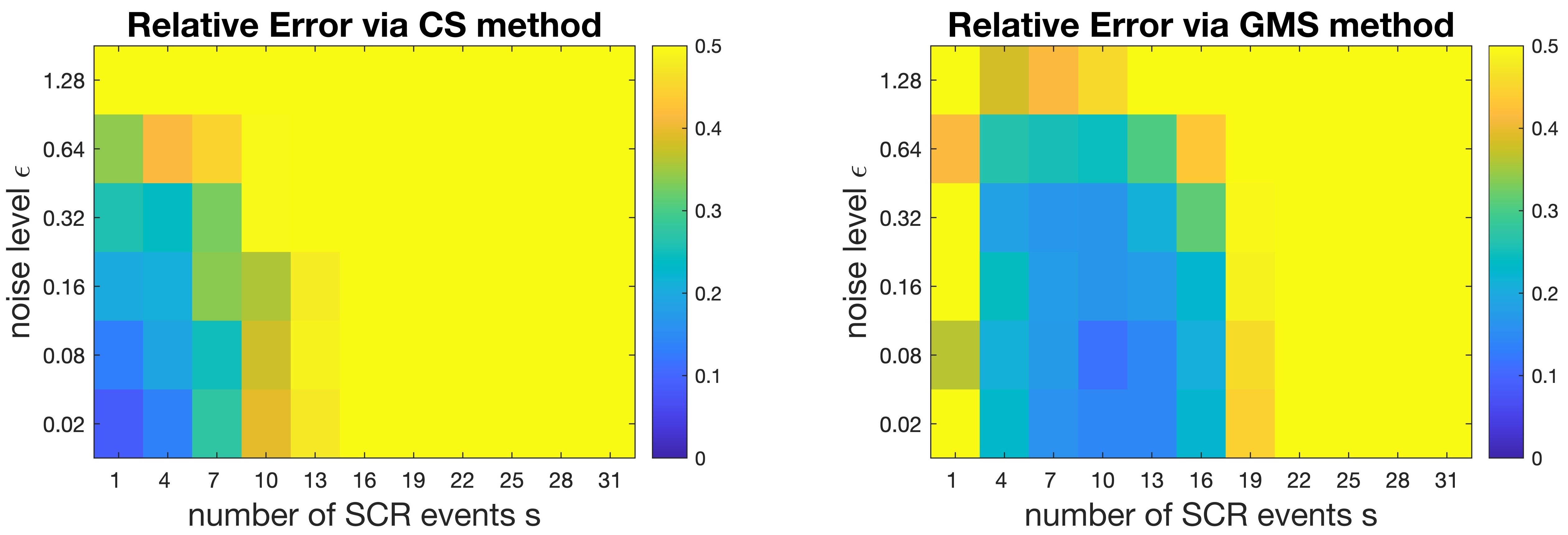}
\caption{Mean relative error for different values of $s$ and $\epsilon$ with $\vx$ and $\vb$ follow XU-BJ combination}\label{fig:heat2}
\end{figure}

Figure \ref{fig:heat4} shows heatmaps when $\vy$ follows the XE-BJ model with 1 random shift. 
The performance is similar to Figure \ref{fig:heat2} where $\vb$ also follows the BJ model: The CS method produces better recovery for small sparsity values ($s=1,4,7, 10$) whereas our GMS method performs better with bigger sparsity.

\begin{figure}[htbp]
\includegraphics[width=\textwidth]{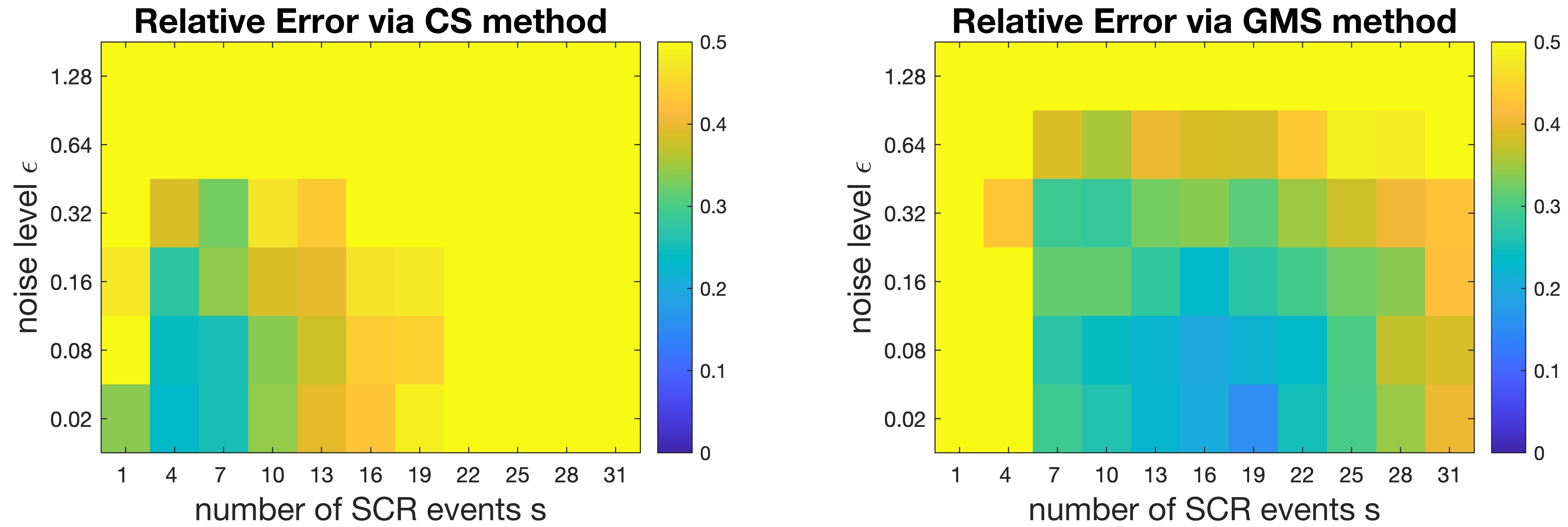}
\caption{Mean relative error for different values of $s$ and $\epsilon$ with $\vx$ and $\vb$ follow XE-BJ combination}\label{fig:heat4}
\end{figure}

\subsubsection{Exploration on effect of $\tau_1, \tau_2$}\label{sec:tau12}
In the previous two experiments, we have  fixed $\mH$ using the parameters $\tau_1=2, \tau_2=0.75$. However, the recovery performance may rely on the filter $\mH$ as well. In this experiment, we choose 5 different values of $\tau_1\in\{2,4,6,8,10\}$ and 3 different values of $\tau_2\in\{0.5,0.75,1\}$, which creates 15 different combinations of $\tau_1$-$\tau_2$ (therefore 15 different $\mH$'s). We have $n=240, K=40, s=10, \delta=10, \gamma=10, \epsilon=0.3$ fixed.

For each $\tau_1$-$\tau_2$ combination, we randomly generate $K=40$ instances of $\vx, \vb, \ve$ (and therefore $\vy$) following Table \ref{tab:random}. We then compute the mean relative error on recovering $\vx$, averaged over these 40 trials by either the CS method or the GMS method. Table \ref{tab:tau} shows the relative errors using the XU-BJ (1 random jump/shift) model. 

\begin{table}[htb]
\caption{Relative Errors for 15 different $\mH$'s}\label{tab:tau}

\noindent\begin{tabular}{cc|ccccc|}
\multicolumn{7}{c|}{XU-BJ model: Relative error by \textbf{CS} method}\\\hline
&&\multicolumn{5}{c|}{$\tau_1$}\\
&&2&4&6&8&10\\\hline
\multirow{3}{*}{$\tau_2$}&0.5&0.194  &  0.105 &   0.103 &   0.096  &  0.119\\
&0.75&   0.553   &   0.187   &   0.132  &    0.115    &  0.124\\
 &1&   1.000   &   0.337    &  0.235   &   0.179   &   0.145
    \end{tabular}
    \begin{tabular}{|cc|cccccc}
\multicolumn{7}{|c}{XU-BJ model: Relative error by \textbf{GMS} method}\\\hline
&&\multicolumn{5}{c}{$\tau_1$}\\
&&2&4&6&8&10\\\hline
\multirow{3}{*}{$\tau_2$}&0.5&0.166  &  0.143 &  0.137  &  0.132  &  0.126\\
&0.75&    0.246 &   0.225  &  0.195  &  0.186 &  0.201\\
 &1&   0.399  &  0.290  &  0.276  &  0.252 &  0.240
    \end{tabular}
\end{table}

We see that for the XU-BJ signal model, 
\begin{itemize}
\item For $\tau_2=0.5$, both methods perform well with the CS method slightly outperforms.
\item For $\tau_2=0.75$, both methods perform similarly but the CS method is very poor for $\tau_1=2$.
\item For $\tau_2=1$, both methods perform similarly but the CS method is very poor for $\tau_1=2$.
\end{itemize} 
Overall, for the XU-BJ model where the CS method may have an advantage, the GMS method still performs better on average (right hand side of Figure \ref{fig:tau12}). 
More importantly, the performance by the GMS method is much more stable with respect to the filter. This is a desirable feature as the parameter of $H$ can vary.

We tested the other three signal models as well. To best visualize the results, Figure \ref{fig:tau12} displays the boxplots and summary statistics of all 4 models by both methods. For each method and each signal model, we draw a boxplot of the 15 relative errors to best visualize them. This is the 8 boxplots on the left. To supplement each boxplot, we have supplied the mean and standard deviation on the right. The GMS method is performing better for all the signal models except for XE-BJ. However, the performance is nearly the same while the GMS method has a much smaller standard deviation, indicating more stable performance with respect to the filter $\mH$.

\begin{figure}[htb]
\begin{minipage}[t]{0.47\textwidth}
\vspace{0pt}
\vspace{-3ex}
\includegraphics[width=1\textwidth]{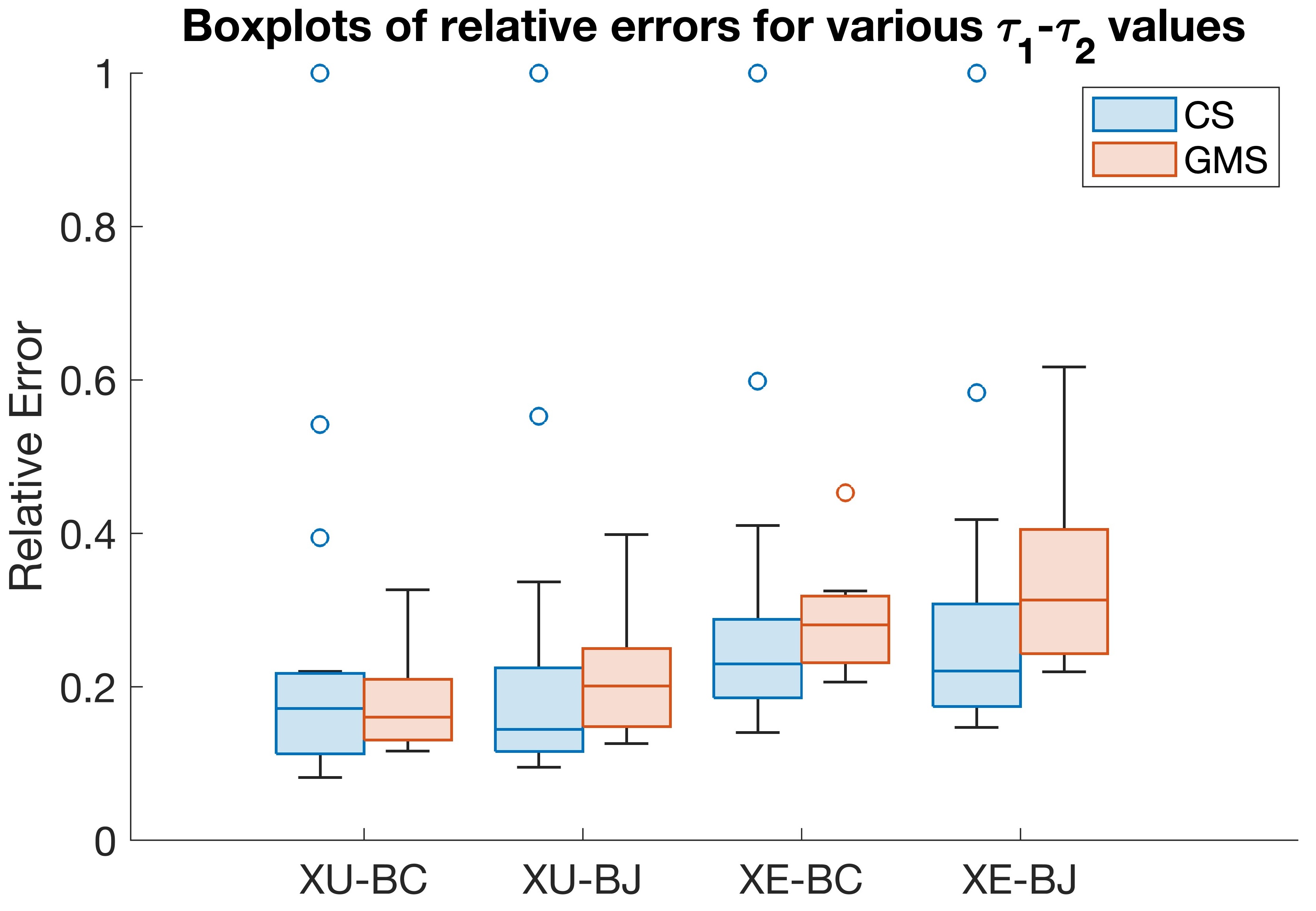}
\end{minipage}\quad
\begin{minipage}[t]{0.5\textwidth}
Mean of relative error over all 15 $\tau_1$-$\tau_2$ values

\vspace{0.05in}
\begin{tabular}{c|cccc}
\hline
mean&XU-BC    &  XU-BJ   &   XE-BC  &    XE-BJ \\
\hline
CS &0.244    &0.241    &0.305    &\textbf{0.299}\\
GMS   &  \textbf{0.180}   &  \textbf{0.214}  &  \textbf{0.281}  &  0.339\\
\hline
\end{tabular}
\vspace{0.2in}

St. dev. of relative error over all 15 $\tau_1$-$\tau_2$ values

\vspace{0.05in}
\begin{tabular}{c|cccc}
\hline
sd&XU-BC    &  XU-BJ   &   XE-BC  &    XE-BJ \\\hline
CS   &  0.243   & 0.241  &  0.224  &  0.227\\
    GMS  &  \textbf{0.057}   & \textbf{0.074}  &  \textbf{0.065}  &  \textbf{0.115}\\\hline
\end{tabular}
\end{minipage}
\caption{Summary graphs and statistics of relative errors of 4 different signals models via the CS or the GMS method. For example, the pair of boxplots of XU-BJ are drawn using the values in Table \ref{tab:tau}.}\label{fig:tau12}
\end{figure}

\subsubsection{Experiments on joint recovery}\label{sec:syn_p}
This experiment focuses on joint/paralleled signal recovery after EDA signal reshaping. It also serves as an exploration of parameter choices for the real data experiments in the next section. The filter $\mH$ is generated with $\tau_1=2, \tau_2=0.75$.

We randomly generate $K$ EDA signals  following the XU-BC model. The parameters are $s=20, \delta=10, \gamma=10, \epsilon=0.3$ with each signal's length to be $N=1360$. This is meant to mimic one subject's EDA signal from real data experiments (See Section \ref{sec:real}). We let $\mY_0$ be the $N\times K$ EDA signal matrix.  We then reshaped each EDA signal using $C=5, q=0.8$, so each $1360\times1$ signal is reshaped into a $272\times 21$ matrix. We use three different methods to recover the SCR signal events:
\begin{itemize}
\item CS method: use each column of $\mY_0$ as input $\vy_0$ in \eqref{equ:lasso}. Done $K$ times.
\item CS-P method: see \eqref{equ:csp}.
\item GMS-P method: see \eqref{equ:gmsp}
\end{itemize}

For each $K\in\{1, 2, 3, 4\}$, we run 40 trials of the above described experiment using all these 3 methods. The mean relative error and run time are displayed in Table \ref{tab:pr}. Among 4 different $K$ values, the GMS-P method has the best accuracy except for $K=1$. The recovery accuracy for GMS-P increases as $K$ increases, demonstrating the power of our joint/paralleled recovery.

\begin{table}[htbp]
\centering
\caption{Performance Comparison Among Three Methods}\label{tab:pr}
\begin{tabular}{cccc}
\multicolumn{4}{c}{Mean Relative Error over 40 trials}\\\hline
&CS&CS-P&GMS-P\\
$K=1$ &    \textbf{0.168}   &  0.179  &  0.195\\
 $K=2$ &    0.165 &   0.180  &  \textbf{0.136}\\
 $K=3$& 0.178&0.190&   \textbf{0.121}\\
 $K=4$& 0.173&    0.187&    \textbf{0.122}
\end{tabular}
\hspace{1cm}
\begin{tabular}{cccc}
\multicolumn{4}{c}{Mean Run Time (seconds) over 40 trials}\\\hline
&CS&CS-P&GMS-P\\
$K=1$ & 2.38 &  0.18 &  \textbf{0.17}\\
$K=2$ & 2.93  &  0.47 &  \textbf{0.26}\\
$K=3$ & 2.70    &\textbf{0.32}&   0.43\\
$K=4$ &2.63 &   \textbf{0.36} &   0.50
\end{tabular}
\end{table}

It is worth noting that although CS and CS-P have similar accuracy, the overlapped reshape (CS-P) is a lot more computationally efficient.

This last synthetic data experiment also paves the way for our real data experiments as it informs a viable choice for $C,q$ in the overlapped reshape transformation.

\subsection{Experiments on Real Data}\label{sec:real}

Thirty volunteers were recruited and  asked to watch a video that lasts 5 minutes and 26 seconds. Among the 30 volunteers, 28 subjects' data was collected. Figure \ref{fig:EDA1markers} shows the EDA signal from one subject with video event markers (see Table \ref{tab:vid}). More details of this experiment can be found in Appendix \ref{sec:exp_detail}.
\begin{figure}[htbp]
\centering
\includegraphics[width=1.\textwidth]{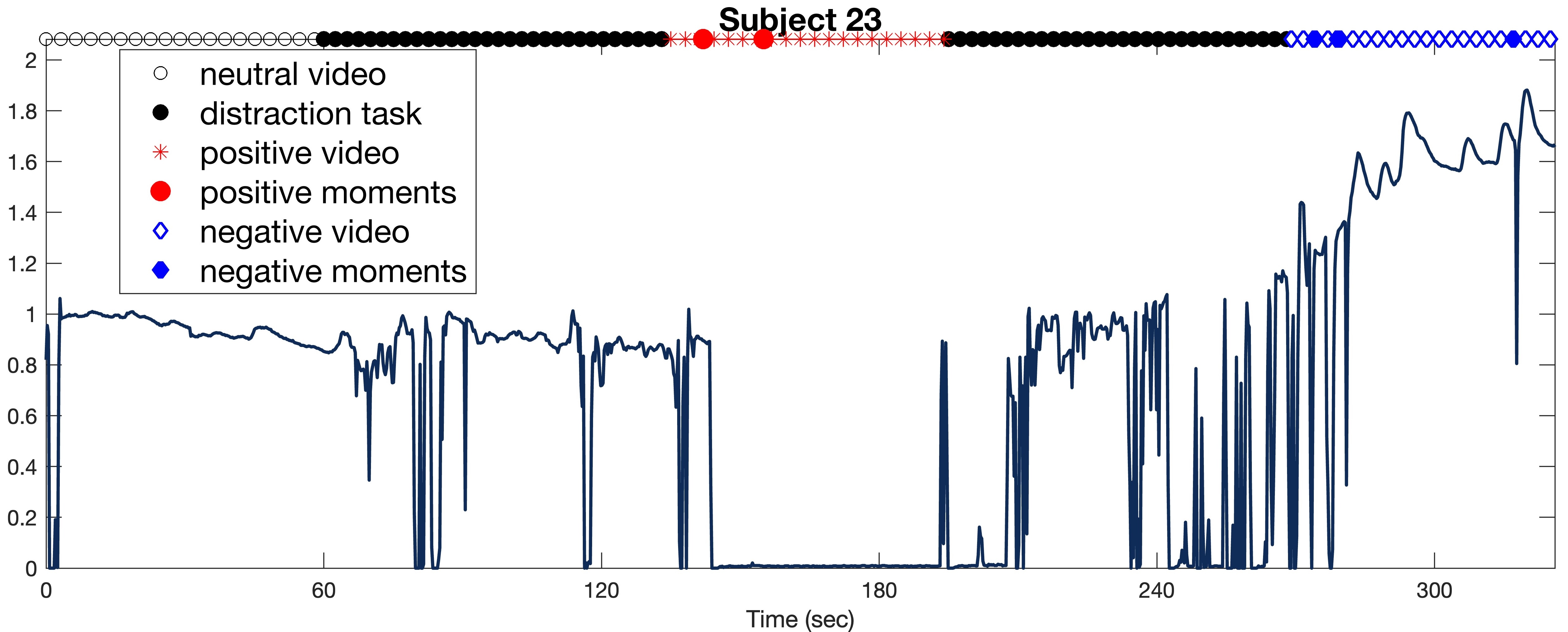}
\caption{EDA signal of Participant 23 with video events markers}\label{fig:EDA1markers}
\end{figure}

There are 5 participants' data that are not usable (almost constantly 0).  In the end, we consider 23 subjects' EDA signal. Each subject's data is trimmed such that the video starts at 7th second.  The signals are sampled at 4Hz so the signal length for each subject is $N=4*(5*60+32) = 1328$. We let $\mY_0=\begin{bmatrix}\vy_1&\vy_2&\cdots&\vy_{23}\end{bmatrix}\in\R^{1328\times 23}$ be the matrix representing this data set where each column is the EDA signal of a subject.
Figure \ref{fig:eda3by3} shows the raw EDA signal of the first 9 subjects. \textbf{For all the figures in this section, time 0 is the start of the video so the first 6 seconds are not plotted.}

\begin{figure}[htb]
\includegraphics[width=0.95\textwidth]{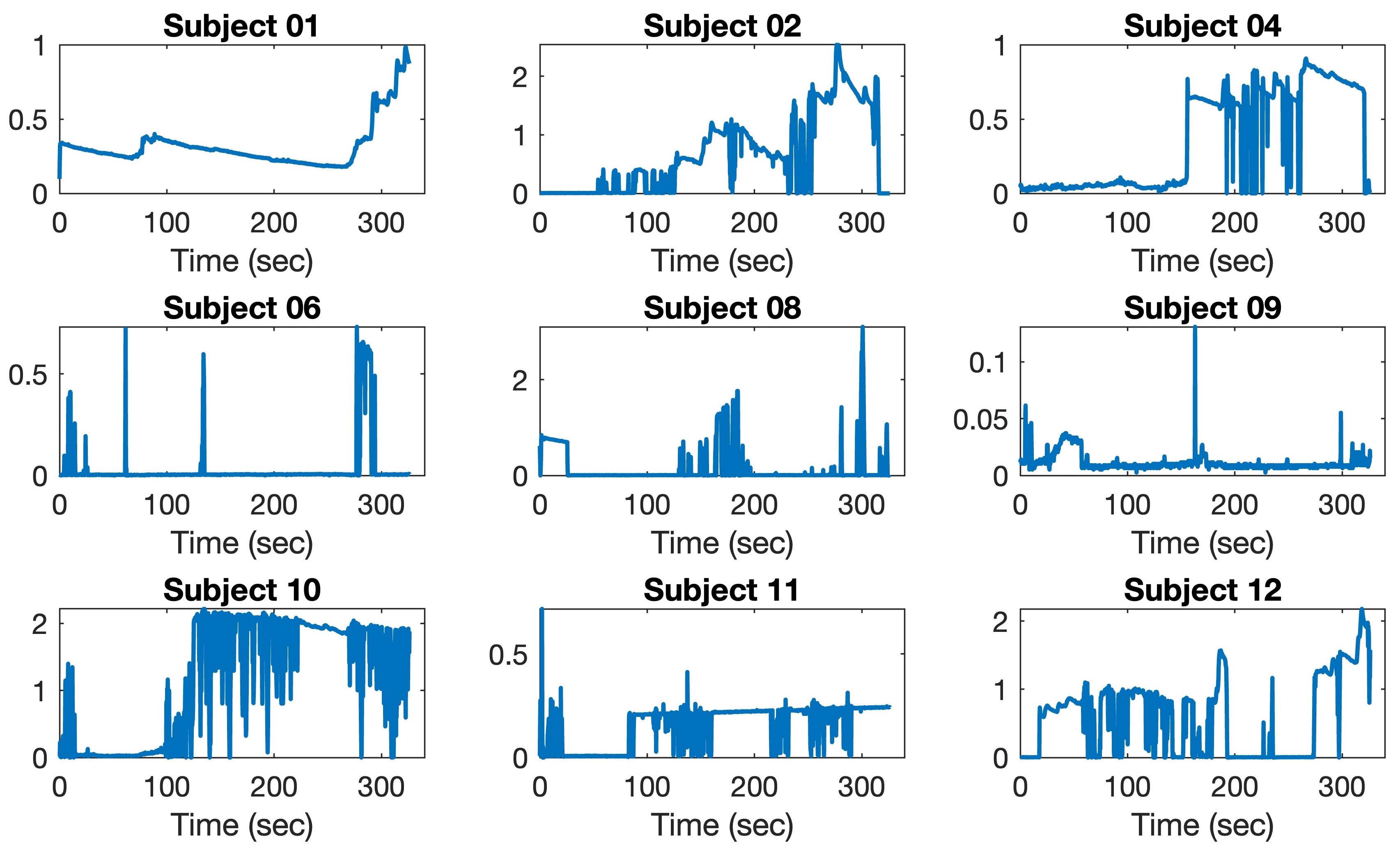}
\caption{EDA signals of 9 participants. }
\label{fig:eda3by3}
\end{figure}

We attempt to recover the emotional events by five different methods:
\begin{itemize}
\item GMS-P method: use  \eqref{equ:gmsp} with $C=5, q=0.85$. This results $\OR_{C,q}(\vy_i)\in\R^{266\times28}$ and consequently $\OR_{C,q}(\mY_0)\in\R^{266\times (28*23)}$.

\item CS method: for $i=[23]$, use each $\vy_i$  as input $\vy_0$ in \eqref{equ:lasso}. 
\item CS-P method: use \eqref{equ:csp} with $C=5, q=0.85$.
\item NeuroKit2: for $i=[23]$, use each $\vy_i$  as input in the function \verb"eda_process()".
\item sparsEDA: for $i=[23]$, use each $\vy_i$  as input in the Matlab function \verb"sparsEDA()".\footnote{https://github.com/fhernandogallego/sparsEDA}
\end{itemize}
\begin{figure}[htb]
\centering
\includegraphics[width=\textwidth]{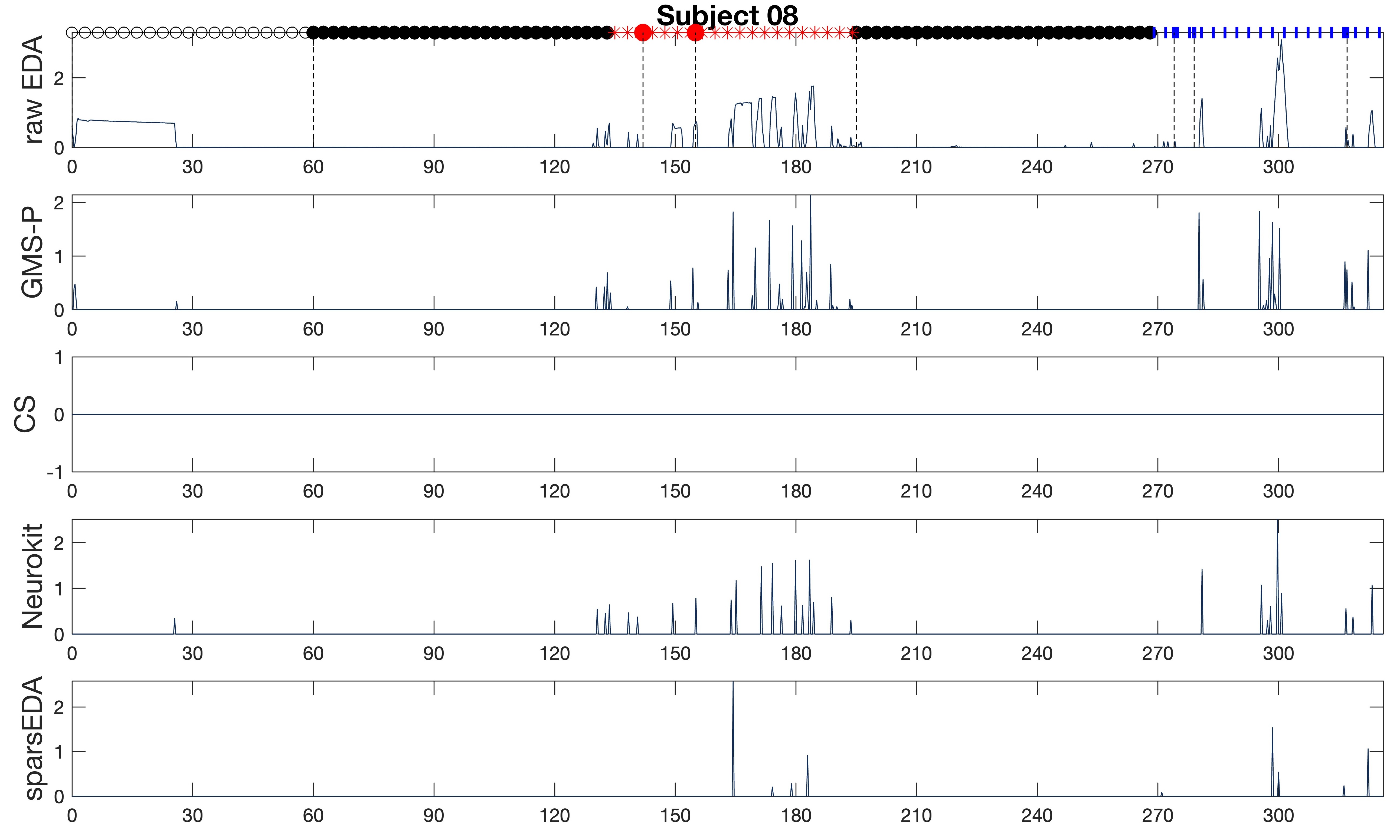}
\caption{Recovery results using all 4 methods}\label{fig:08}
\end{figure}

\begin{figure}[htb]
\centering
\includegraphics[width=\textwidth]{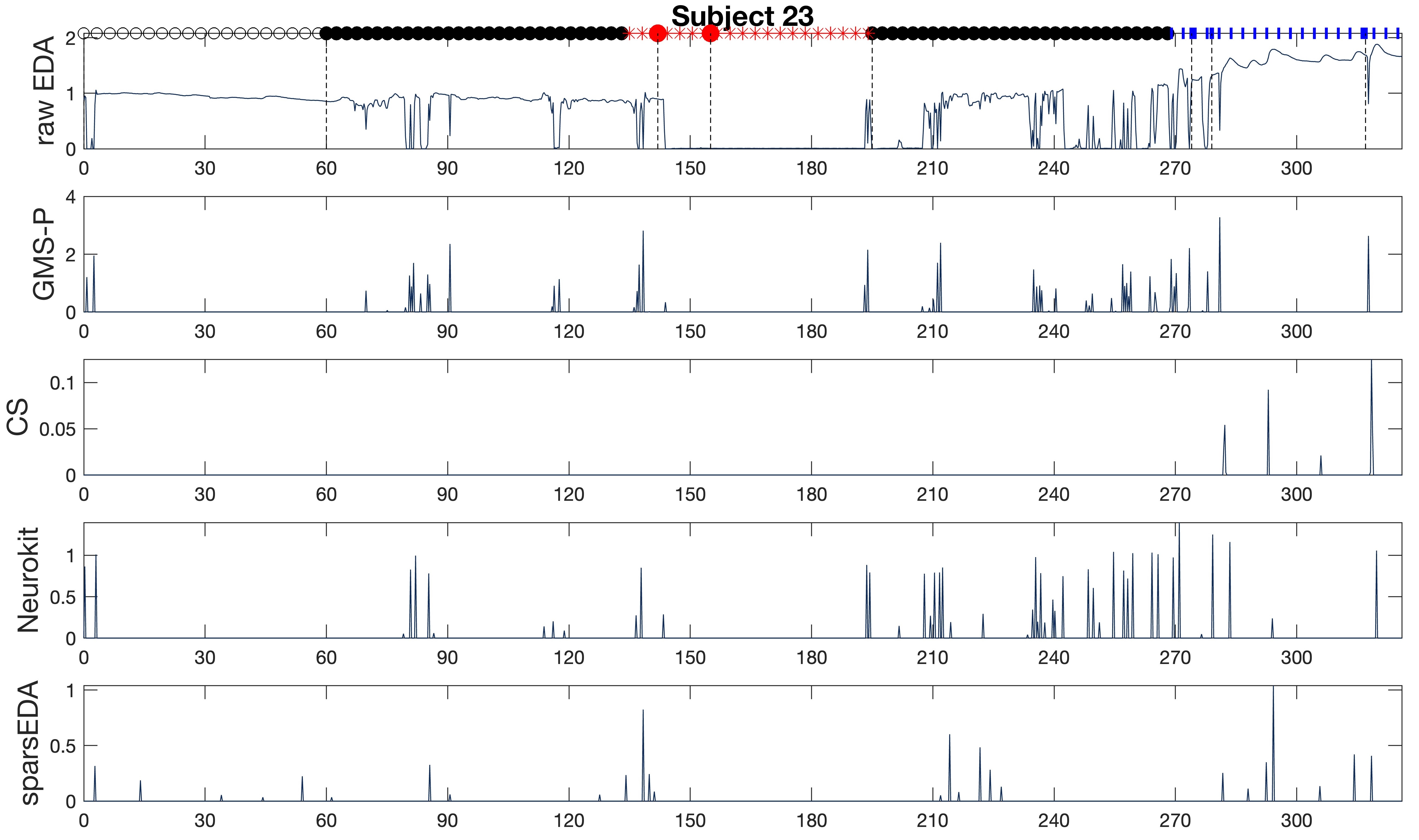}
\caption{Recovery results using all 4 methods}\label{fig:23}
\end{figure}

Figure \ref{fig:08} and Figure \ref{fig:23} display the recovery results of Subject 08 and Subject 23 respectively, using these methods. We observe that:
\begin{itemize}
\item The results from CS and CS-P are almost identical throughout all participants, so we only plot the CS method in Figures \ref{fig:08} and \ref{fig:23}. The amplitudes of SCR obtained are too small in general. Both the CS and CS-P methods are not very successful at identifying stimulus events as the recovered amplitudes of SCR are almost all 0's.\footnote{This was tested with multiple choices of $\lambda$ in \eqref{equ:lasso}.}
\item Our method (GMS-P) has similar results to NeuroKit2 in general, but NeuroKit2 tends to identify more false peaks, as shown in both Figure \ref{fig:08} and Figure \ref{fig:23}.
\item sparsEDA detects fewer peaks in general.
\end{itemize}

To compare all these different methods in a more systematic and quantitative way, we define \emph{event match rate (ER)} and \emph{false peak rate (FR)} given an onset time window $t$. For any recovered SCR signal, a peak is identified if its amplitude is at least 2\% of the maximum value of the EDA signal.
\begin{align*}
\text{event match rate}=\frac{\text{number of identified peaks that are within $t$ seconds of any event}}{\text{number of events}},\\
\text{false peak rate}=\frac{\text{number of identified peaks that are not within $t$ seconds of any event}}{\text{number of peaks}}.
\end{align*}
The events are the beginning of video (0 sec), beginning of two distraction tasks (60, 195 sec), positive moments (142, 155 sec), and negative moments (274,   279,   317 sec).
It is apparent that a good recovery method should have high ER and low FR.

For example, in Figure \ref{fig:23}, given $t=2$, for the CS method, there are two peaks within 2 seconds of any event (last two negative moments), so the event match rate is 2/8 = 0.25. The other two peaks are not in a 2 seconds window of any event, so the false peak rate  is 2/4 = 0.5.

We also used ER to define \emph{quality subjects}.
Among 23 subjects whose EDA signals are processed, we only keep subjects whose ER are greater than or equal to 50\% for at least one of the 5 methods when $t=2$. We call these quality subjects. We ended up with 9 quality subjects.

Table \ref{tab:rates} lists the two rates averaged over all 9 quality subjects for all these methods.
For the event match rate, GMS-P performs the best for 4 out of 5 different window size $t$. The CS method performs poorly for ER so our false peak rate only included GMS-P, NeuroKit2, and SparsEDA, among which GMS-P has the lowest FR.

\begin{table}[htbp]
\caption{ER and FR averaged over all quality subjects for various methods}\label{tab:rates}
\centering
\begin{tabular}{c| c c c c c}
window $t$&0.5 sec &1 sec &1.5 sec & 2 sec &2.5 sec\\
\hline\hline
\multicolumn{6}{c}{(ER) event match rate}\\
\hline
GMS-P & 0.2917  &  \textbf{0.5139}  &  \textbf{0.5278}  &  \textbf{0.6111}   & \textbf{0.6528}\\
 NeuroKit2&   \textbf{0.3472}   & 0.4306  &  \textbf{0.5278}  &  0.5556  &  0.6389\\
 SparsEDA  & 0.1250  &  0.2083  &  0.2500  &  0.2778  &  0.3056\\
  CS&  0.0417  &  0.0417 &   0.0556  &  0.0694   & 0.0694\\
  \hline
\multicolumn{6}{c}{(FR) false peak rate}\\
GMS-P & \textbf{0.9249}  &  \textbf{0.8750}  &  \textbf{0.8464}  &  \textbf{0.7997}  &  \textbf{0.7713}\\
NeuroKit2 &    0.9517  &  0.9207  &  0.8815  &  0.8671  &  0.8423\\
 SparsEDA &   0.9274  &  0.8781  &  0.8610  &  0.8219  &  0.8032
\end{tabular}
\end{table}

\section{Discussion and Conclusion}\label{sec:dis}
We propose a new method gmsEDA that is based on generalized matrix separation~\cite{CW25, CD25} for decomposing EDA signals into its phasic (SCR events) and tonic components. Whether we have one single EDA signal or multiple signals, overlapped reshape is recommended for better performance and  computational efficiency. 
Numerical experiments in Section \ref{sec:syn_p} do suggest that gmsEDA shows an advantage when  at least two EDA signals are being processed simultaneously.

Our setup and modeling can handle shifts in baseline  which comes from motion artifacts. We also provide theoretical guarantee for our matrix recovery problem tailored to the EDA decomposition problem.
Our synthetic experiments further demonstrate that gmsEDA is robust with respect to the choice of $\mH$ which may be changing over time or likely subject dependent.

Compared to other popular toolkits, the real data experiments demonstrate superior performance of gmsEDA in terms of higher event detection rate and lower false peak rate.

Our GMS framework is more general than EDA decomposition. In the future, we will explore applications to imaging and other physiological data, including developing generalized matrix separation theory when the sparse matrix is confined to be nonnegative.


%

\appendix
\section{Mathematical Background}\label{sec:math}
Given a matrix $\mA=(a_{ij})$, 
\begin{itemize}
\item $\|\mA\|$ is its spectral norm
\item $\|\mA\|_\infty=\max_{i,j}|a_{ij}|$
\item $\|\mA\|_F=\sqrt{\sum_{i,j}a_{ij}^2}$ is its Frobenius norm.
\end{itemize}
\begin{definition}[Singular Value Decomposition (SVD)]\label{def:svd}
Let $\mA\in\R^{m\times n}$ whose rank is $k$. It can be shown that $\mA$ can be factorized as
\begin{equation}\label{equ:SVD}
\mA=\mU\Sigma \mV^\top,
\end{equation}
where $\mU$ is an $m\times k$  matrix with orthonormal columns, $\mV$ is a $k\times n$ matrix with orthonormal columns, and $\Sigma=\diag(\sigma_1, \sigma_2, \cdots, \sigma_k)$ is a $k\times k$ diagonal matrix with positive diagonals. 
These positive diagonal entries of $\Sigma$ are called the \emph{singular values} of $\mA$, and can be arranged in descending order: $\sigma_1\geq\sigma_2\geq\cdots\geq \sigma_{k}>0$. The decomposition \eqref{equ:SVD} is called the reduced singular value decomposition of $\mA$.
\end{definition}

With Definition \ref{def:svd}, we further define the condition number of $\mA$ to be $\sigma_1/\sigma_k$. For any $q\in[k]$,  the rank-$q$ approximation of $\mA$, denoted by $\mA_q$ is
$$\mA_q:=\mU_q\Sigma_q\mV_q^\top,$$
where $\mU_q$ consists of the first $q$ columns of $\mU$, $\mV_q$ consists of the first $q$ columns of $\mV$, and $\Sigma_q$ is the leading $q\times q$ submatrix of $\Sigma$.

$\mA_q$ is considered the best rank-$q$ approximation of $\mA$ in terms of both spectral norm and Frobenius norm, that is
$$\mA_q=\argmin_{\rank(B)\leq q}\|\mA-\mB\|, \text{ and }\mA_q=\argmin_{\rank(\mB)\leq q}\|\mA-\mB\|_F.$$
A matrix $\mA$ is often considered ``low-rank'' if $\|\mA-\mA_q\|$ is relatively small for $q$ much less than the number of rows of $\mA$ and the number of columns of $\mA$.

\begin{reading}
Lectures 1-5 of \cite{NLA} is a great reference for SVD and related concepts. For matrix separation, interested readers can refer \cite{C11} and \cite{rpca}.
\end{reading}

\section{Theoretical Guarantee}\label{sec:theory}
We first review some related results for recovery guarantee of \eqref{equ:gms} or \eqref{equ:gmsc}. The following definitions are from \cite{CW25}.

Given a matrix $\mS_0\in \R^{p\times n}$ and $0<\delta <1$, we say a matrix $\mG$ of dimension $m\times p$ has the $\mS_0$-$\delta$-\emph{restricted infinity norm property} ($\mS_0$-$\delta$-RINP) if 
\begin{equation}\label{equ:Hrip}
\|(\mI-\mG^\top \mG)\mA\|_\infty\leq\delta\|\mA\|_\infty \text{ for all }\mA\in\Omega(\mS_0),
\end{equation}
where $\Omega(\mS_0)$ is the set of all $p\times n$ matrices whose support is within the support of $\mS_0$.

For a fixed $\mG$, we also define
\begin{equation}\label{equ:muH}
\mu_\mG(\mS):=\max_{\mA\in\Omega(\mS), \|\mA\|_\infty\leq1}\|\mG\mA\|.
\end{equation}
and
\begin{equation}\label{equ:xiH}
\xi_\mG(\mL):=\max_{\mB\in \Tc(\mL), \|\mB\|\leq1}\|\mG^\top \mB\|_\infty,
\end{equation}
where $\Tc(\mL)$ is  the tangent space at matrix $\mL$ with respect to the variety of all matrices with rank less than or
equal to rank$(\mL)$.

\begin{theorem}[{\cite[Theorem 2.7]{CW25}}]\label{thm:main}\label{thm:main}
Given $\mM_0=\mG\mS_0+\mL_0$ where  $\mG$ satisfies  \eqref{equ:Hrip} with $0\leq\delta<1/3$. If
\begin{equation}\label{equ:incoherence}
\mu_\mG(\mS_0)\xi_\mG(\mL_0)<\frac{1-3\delta}{6},
\end{equation}
then there exists $\lambda>0$ such that for any optimizer $(\hat \mL, \hat \mS)$ of 
\begin{equation}\label{equ:p}
(\hat \mL,  \hat \mS)=\argmin_{\mL, \mS}\{\|\mL\|_*+\lambda \|\mS\|_1\} \quad \text{subject to }\mM_0=\mL+\mG\mX,
\end{equation}
we must have $\hat \mS=\mS_0, \hat \mL=\mL_0$.  
\end{theorem}

\begin{lemma}\label{lem}
Problem \eqref{equ:gmsc} is equivalent to 
\begin{equation}\label{equ:p2_equiv}
( \tilde{\mL}, \tilde{\mX}) = \underset{\mL, \mX}{\argmin} \;\; \lambda \|\mX\|_1 + \|\mL\|_*, \quad\text{subject to } \Sigma_\mH^{-1} \mU_\mH^\top \mY_0 = \mV_\mH^\top \mX + \mL
\end{equation}
with the correspondence $\hat{\mX}_c= \tilde{\mX}, \hat\mW_c=\mU_\mH\tilde\mL$ (or $\tilde{\mL}=\mU_\mH^\top\hat\mW_c$).
\end{lemma}
\begin{proof}
Given $(\hat{\mW}_c, \hat{\mX}_c)$ from \eqref{equ:gmsc}, 
we first show that $(\mU_\mH^\top\hat\mW_c, \hat{\mX}_c)$ solves \eqref{equ:p2_equiv}.

Multiplying both sides of the constraint equation $\mC \mY_0=\hat\mW_c + \mC\mH\hat\mX_c$ by $\mU_\mH^\top$ yields $\Sigma_\mH^{-1} \mU_\mH^\top \mY_0 = \mV_\mH^\top \hat\mX_c + \mU_\mH^\top\hat\mW_c$, showing that $(\mU_\mH^\top\hat\mW_c, \hat{\mX}_c)$ is feasible in \eqref{equ:p2_equiv}.

For any $(\mL, \mX)$ feasible in \eqref{equ:p2_equiv}, we have $( \mU_\mH\mL, \mX)$  feasible in \eqref{equ:gmsc}, which produces a smaller or equal objective function value than the optimizer $(\hat{\mW}_c, \hat{\mX}_c)$:
\begin{equation}\label{equ:a1}
\lambda \|\mX\|_1 + \|\mU_\mH\mL\|_*\geq \lambda \|\hat\mX_c\|_1 + \|\hat\mW_c\|_*.
\end{equation}

 We have $\|\mU_\mH^\top\hat\mW_c\|_*=\|\mU\mU_\mH^\top\hat\mW_c\|_*=\|\hat\mW_c\|_*$ given $\hat\mW_c$ is in the range of $\mU_\mH$. So
 $$\lambda \|\hat\mX_c\|_1 + \|\mU_\mH^\top\hat\mW_c\|_*=\lambda \|\hat\mX_c\|_1+\|\hat\mW_c\|_*\stackrel{\eqref{equ:a1}}{\leq}\lambda \|\mX\|_1 + \|\mU_\mH\mL\|_*=\lambda \|\mX\|_1 + \|\mL\|_*$$
 showing that $(\mU_\mH^\top\hat\mW_c, \hat{\mX}_c)$ solves \eqref{equ:p2_equiv}. The other direction is similar.
\end{proof}

Let $\mP_{\ker(\mH)}$ denote the orthogonal projection onto the kernel of $\mH$.
\begin{theorem}\label{thm:prec}
Given constraint $\mY_0=\mB_0+\mH\mX_0$ and $\mH=\mU_\mH\mathbf{\Sigma}_\mH\mV_\mH^\top$ be a reduced SVD. If there exists $0\leq\delta<1/3$ such that
\begin{align}
\label{equ:Hrip2}
\|\mP_{\ker(\mH)}\mA\|_\infty\leq\delta\|\mA\|_\infty \text{ for all }\mA\in\Omega(\mX_0),\\
\label{equ:incoherence2}
\mu_{\mV^\top_\mH}(\mX_0)\xi_{\mV^\top_\mH}(\mB_0)<\frac{1-3\delta}{6},
\end{align}
then there exists $\lambda>0$ such that for any optimizer $(\hat \mB_c, \hat \mX_c)$ of \eqref{equ:gmsc}, we must have $\hat \mX_c=\mX_0, \hat \mB_c=\mB_0$.  
\end{theorem}
\begin{proof}
Given $(\hat{\mW}_c, \hat{\mX}_c)$ from \eqref{equ:gmsc}, we have $(\mU_\mH^\top\hat\mW_c, \hat{\mX}_c)$ solves \eqref{equ:p2_equiv}.

We see that $\mV^\top_\mH$ satisfies $\mX_0$-$\delta$-RINP since $\mI-\mG^\top \mG=\mI-\mV_\mH\mV^\top_\mH=\mP_{\ker(\mH)}$. With the constraint $\Sigma_\mH^{-1} \mU_\mH^\top \mY_0 = \mV_\mH^\top \mX_0 + \Sigma_\mH^{-1} \mU_\mH^\top\mB_0$,  we can apply Theorem \ref{thm:main} to the problem \eqref{equ:p2_equiv} where $\mG$ is $\mV^\top_\mH$. We conclude that $\hat\mX_c=\mX_0, \mU_\mH^\top\hat\mW_c=\Sigma_\mH^{-1} \mU_\mH^\top\mB_0$.

Finally, $\hat \mB_c=\mY_0 - \mH \hat{\mX}_c=\mY_0-\mH\mX_0=\mB_0$.
\end{proof}

We analyze \eqref{equ:Hrip2}  with our filter $\mH$ as defined in \eqref{equ:H}. $\mH$ is lower triangular with the diagonals being $\vh_1=f(0)=0$. It is easy to check that $\rank(\mH)=n-1$ and $\ker(\mH)=\text{span}\{(0, 0, \cdots, 1)\}$. \eqref{equ:Hrip2} can be satisfied with $\delta=0$ if the last row of $\mX_0$ are all 0. In terms of EDA signal, this translates to no stimulus event occurs at the very last entry. This is easy to achieve especially with our overlapped reshaping technique: if there is positive entry (stimulus event) at the last entry of a segment, it will resurface in the next segment (subsequent column) in the middle position and hence detected.

For \eqref{equ:incoherence2}, the analysis can follow from \cite[Sec 3.1]{CW25} and \cite[Prop 3, Prop 4]{C11} especially given that $\mV_\mH$ has codimension 1. However, the simulated results often perform better than what \eqref{equ:incoherence2} allows in terms of $\rank(\mB_0)$ and sparsity of $\mX_0$.

\section{Details of the Real Data Experiment}\label{sec:exp_detail}

\subsection{Participants}
Thirty undergraduate students were recruited from a mid-sized southern university and completed the study in a laboratory setting in exchange for course credit. 

\subsection{Design}
The study employed a within-subjects repeated-measures mood induction design in which participants completed neutral, positive, and negative mood conditions. Mood condition served as the within-subjects independent variable, and affective responses were assessed using self-report measures (PANAS) and physiological indices collected via the Empatica E4 wristband. Conditions were presented in a fixed order (neutral, positive, negative) to prevent carryover effects from the negative affect manipulation.

\subsection{Equipment}
Physiological data was collected using the Empatica E4 wristband (Empatica Inc., Cambridge, MA), a wearable device designed for physiological data collection during the study. The E4 records EDA via sensors in contact with the participant’s wrist, allowing for assessment of sympathetic nervous system activity across affect manipulations. Prior research has demonstrated the utility of EDA measures obtained from the E4 in distinguishing periods of rest from periods of increased physiological arousal~\cite{schuurmans2020validity}, as well as sensitivity to changes in positive and negative affective states~\cite{borrego2019reliability}.


\subsection{Mood Induction Stimuli}
The mood induction stimuli were categorized into neutral, positive, and negative conditions. Neutral stimuli consisted of two clips with a combined duration of 58 seconds depicting passive, everyday situations. 
Clips were presented in a fixed order (neutral, positive, negative) to minimize affective carryover between conditions. The negative condition was presented last to reduce potential contamination of physiological responses during the neutral and positive conditions. 
See Table~\ref{tab:vid} for the video timestamps.
\begin{table}[htb]
\caption{Video Breakdown}\label{tab:vid}
\begin{tabular}{ll}
0:00 – 0:28:  &Neutral video, scene of people walking on a busy city street\\
0:29 – 0:59:  &Neutral video, scene of two people working together on a laptop\\
1:00 – 2:14:  &Black screen with text ``Please complete distractor task now''\\
2:15 – 2:42:  &Positive video, baby attempting but failing to drink water from a hose\\
	&2:22: first laugh moment\\
	&2:35: Baby smiles\\
2:43 – 3:14:  &Positive video, cat staring at camera, wiggling tongue\\
	&Video is consistent throughout (no unique moments)\\
3:15 – 4:28:  &Black screen with text ``Please complete distractor task now''\\
4:29 – 4:57:  &Negative video, skateboarder falling and breaking arm\\
	&4:34:  Moment of fall (doesn’t look overly disturbing)\\
	&4:39:  Broken arm clearly shown (very disturbing) \\
4:58 – 5:25:  &Negative video, animal trainer has arm chomped by alligator\\
	&5:17:  Alligator chomps on arm and begins to roll with arm in mouth\\
	\end{tabular}
\end{table}

\subsection{Procedure}
Participants completed a 1-hour laboratory session that began with presession procedures, including review of the informed consent protocol, disclosure of potentially distressing material, and completion of baseline self-report measures  \cite{watson1988development, radloff1977ces, spitzer2006brief}. Following baseline assessment, participants were fitted with the Empatica E4 wristband, the device serial number was recorded, and continuous physiological recording was initiated. The neutral condition was presented first and consisted of two neutral clips shown consecutively. Immediately following the clips, participants completed the PANAS and then engaged in a 1-minute distractor task involving simple arithmetic problems. The positive condition followed the same sequence: two positive clips presented consecutively, PANAS administration, and a 1-minute distractor task. The negative condition was then presented and consisted of two negative clips shown consecutively, followed by PANAS administration. No distractor task was included after the negative condition. Immediately thereafter, participants viewed an additional positive clip to facilitate return to a neutral-to-positive affective state. The E4 device was then removed, and participants completed a post-film questionnaire assessing prior exposure to the clips. Participants were subsequently debriefed, and physiological and survey data were uploaded for analysis.

\section*{Acknowledgments}
The authors thank OpiAID\footnote{https://opiaid.ai/} for providing Empatica E4 for data collection. Additionally, we would like to thank the student research assistants that helped run experimental sessions, Peyton Farmer-Twiddy and Daphne Kilbourne.  


\section*{Access to Code}
We created a github repository \url{https://github.com/xuemeic/gmsEDA} to make the real data available, as well as the  code for our gmsEDA method.
\bibliographystyle{siamplain}
\bibliography{ref_eda}

\end{document}